\documentclass{theoretics}
\usepackage[T1]{fontenc}
\usepackage{nico}
\usetikzlibrary{shapes.symbols}

\newcommand{\Call}[2]{\ProcNameSty{#1}(#2)} 
\newcommand{\Calli}[1]{\textsc{#1}} 
\newcommand{\PO}{Player~$1$\xspace}
\newcommand{\PT}{Player~$2$\xspace}
\newcommand{\PLi}{Player~$i$\xspace}

\newcommand{\C}{{\cal C}}

\newcommand{\G}{{\mathcal G}}
\newcommand{\R}{{\mathcal R}}

\renewcommand{\S}{{\mathcal S}}

\newcommand{\fr}{\texttt{fr}}

\newcommand{\Nat}{\mathbb{N}}
\newcommand{\Natstr}{\mathbb{N}^*}
\newcommand{\Natstro}{\mathbb{N}^* \setminus \mathbb{N}}

\newcommand{\notF}{V \setminus F}
\newcommand{\notFd}{V \setminus F_d}

\newcommand{\thresh}{\texttt{Th}\xspace}

\newcommand{\coTh}{\texttt{coB\"u-Th}}

\newcommand{\safe}{\text{Safe}}
\newcommand{\coBuchi}{\text{co-B\"uchi}\xspace}
\newcommand{\frparity}{\text{Fr-Parity}}

\newcommand{\play}{\text{play}\xspace}
\newcommand{\Spare}{\text{Spare}\xspace}

\NewDocumentCommand{\tbid}{gg}{b^{\IfNoValueTF{#2}{T}{#2}}_{\IfNoValueTF{#1}{}{#1}}}
\newcommand{\zug}[1]{\langle #1  \rangle}
\newcommand{\stam}[1]{}
\renewcommand{\set}[1]{\{ #1  \}} 
\newcommand{\succb}[1]{#1 \oplus 0^*\xspace}
\newcommand{\predb}[1]{#1 \ominus 0^*\xspace}
\newcommand{\succInt}[1]{\left(#1\right)^*}
\newcommand{\eps}{\varepsilon}
\DeclarePairedDelimiter{\adjustfloor}{\lfloor}{\rfloor}
\newcommand{\floor}[1]{\adjustfloor*{#1}}
\newcommand{\adjustbrac}[1]{\left(#1\right)}
\newcommand{\Obj}{\mathcal{O}}
\NewDocumentCommand{\ftg}{g}{\tilde{\IfNoValueTF{#1}{f}{#1}}}
\NewDocumentCommand{\configtg}{gg}{\tilde{\IfNoValueTF{#1}{c}{#1}}_{\IfNoValueTF{#2}{}{#2}}}
\NewDocumentCommand{\playtg}{g}{\tilde{\IfNoValueTF{#1}{\pi}{#1}}}
\newcommand{\placeholdr}{\hat{T}(v)}

\NewDocumentCommand{\invT}{g}{\IfNoValueTF{#1}{T^{\prime}}{{#1}^{\prime}}}
\NewDocumentCommand{\vminus}{g}{v^-_{\IfNoValueTF{#1}{}{#1}}}
\NewDocumentCommand{\vplus}{g}{v^+_{\IfNoValueTF{#1}{}{#1}}}
\newcommand{\absolut}[1]{|#1|}
\NewDocumentCommand{\sumT}{g}{\IfNoValueTF{#1}{\absolut{T(\vplus)} + \absolut{T(\vminus)}}{\absolut{#1(\vplus)} + \absolut{#1(\vminus)}}}\NewDocumentCommand{\diffT}{g}{\IfNoValueTF{#1}{\absolut{T(\vplus)} - \absolut{T(\vminus)}}{\absolut{#1(\vplus)} - \absolut{#1(\vminus)}}}
\newtheorem{problem}{Problem}
\newcommand{\half}[1]{\frac{#1}{2}}

\title{Computing Threshold Budgets in Discrete-Bidding Games}

\ThCSauthor{Guy Avni}{gavni@cs.haifa.ac.il}[0000-0001-5588-8287]
\ThCSauthor{Suman Sadhukhan}{suman.sadhukhan00@gmail.com}[0000-0002-4802-6803]
\ThCSaffil{Department of Computer Science, University of Haifa, Israel}
\ThCSthanks{This research was supported in part by ISF grant no. 1679/21. A preliminary version of this article appeared at FSTTCS 22 \cite{AvniS22}.}
\ThCSshortnames{G. Avni, S. Sadhukhan}  
\ThCSshorttitle{Computing Threshold Budgets in Discrete-Bidding Games}
\ThCSyear{2025}
\ThCSarticlenum{5}
\ThCSreceived{Jan 4, 2024}
\ThCSrevised{Sep 27, 2024}
\ThCSaccepted{Nov 23, 2024}
\ThCSpublished{Jan 22, 2025}
\ThCSupdated{Jul 2, 2025}
\ThCSdoicreatedtrue
\ThCSkeywords{Discrete bidding games, Richman games, parity games, reachability games}

\begin{document}
\maketitle

\begin{abstract}
In a two-player zero-sum graph game, the players move a token throughout a graph to produce an infinite play, which determines the winner of the game. \emph{Bidding games} are graph games in which in each turn, an auction (bidding) determines which player moves the token: the players have budgets, and in each turn, both players simultaneously submit bids that do not exceed their available budgets, the higher bidder moves the token, and pays the bid to the lower bidder (called {\em Richman} bidding). We focus on {\em discrete}-bidding games, in which, motivated by practical applications, the granularity of the players' bids is restricted, e.g., bids must be given in cents. 

A central quantity in bidding games is {\em threshold budgets}:  a necessary and sufficient initial budget for winning the game. Previously, thresholds were shown to exist in parity games, but their structure was only understood for reachability games. Moreover, the previously-known algorithms have a worst-case exponential running time for both reachability and parity objectives, and output strategies that use exponential memory. 
We describe two algorithms for finding threshold budgets in parity discrete-bidding games. The first is a fixed-point algorithm. 
It reveals, for the first time, the structure of threshold budgets in parity discrete-bidding games. 
Based on this structure, we develop a second algorithm that shows that the problem of finding threshold budgets is in \NP and co\NP for both reachability and parity objectives. Moreover, our algorithm constructs strategies that use only linear memory.

This is a corrected version of the paper (\url{https://arxiv.org/abs/2210.02773v4}) published originally on Jan 22, 2025.
\end{abstract}

\section{Introduction}
Two-player zero-sum {\em graph games} are a central class of games. A graph game proceeds as follows. A token is placed on a vertex and the players move it throughout the graph to produce an infinite play, which determines the winner of the game. The central algorithmic problem in graph games is to identify the winner and to construct winning strategies. One key application of graph games is {\em reactive synthesis}~\cite{PR89}, in which the goal is to synthesize a {\em reactive system}, namely a policy for interacting with an adversarial environment,  that satisfies a given specification no matter how the environment behaves. 

Two orthogonal classifications of graph games are according to the {\em mode} of moving the token and according to the players' {\em objectives}. For the latter, we focus on two canonical qualitative objectives. In {\em reachability} games, there is a set of target vertices and \PO wins if a target vertex is reached. 
In {\em parity} games, each vertex is labeled with a parity index and an infinite path is winning for \PO iff the highest parity index that is visited infinitely often is odd. The simplest and most studied mode of moving is {\em turn-based}:  the players alternate turns in moving the token. We note that reactive synthesis reduces to solving a turn-based parity game. Examples of other modes of moving are {\em concurrent} and {\em probabilistic} moves (see \cite{AG11}). 

We study \emph{bidding graph games}~\cite{LLPU96,LLPSU99}, which apply the following mode of moving: both players have budgets, and in each turn, an auction (bidding) determines which player moves the token. Concretely, we focus on {\em Richman bidding} (named after David Richman): in each turn, both players simultaneously submit bids that do not exceed their available budget, the higher bidder moves the token, and pays his bid to the lower bidder. Note that the sum of budgets stays constant throughout the game. 
We distinguish between {\em continuous}- and {\em discrete}-bidding, where in the latter, the granularity of the players' bids is restricted. 
The central questions in bidding games revolve around the {\em threshold budgets}, namely a necessary and sufficient initial budget for winning the game.

\subparagraph*{Continuous-bidding games.} 
This paper focuses on discrete-bidding. We briefly survey the relevant literature on continuous-bidding games, which have been more extensively studied than their discrete-bidding counterparts. Bidding games were introduced in~\cite{LLPU96,LLPSU99}. The objective that was considered is a variant of reachability, which we call {\em double reachability}: each player has a target and a player wins if his target is reached (unlike reachability games in which \PT's goal is to prevent \PO from reaching his target). The targets are assumed to be distinct, thus the game is zero sum. Note that apriori, it is possible that a play results in a tie if no target is visited. It was shown, however, that in continuous-bidding games, a target is necessarily reached, thus double-reachability games essentially coincide with reachability games continuous-bidding games. 

Threshold budgets were shown to exist; namely, each vertex $v$ has a value $\thresh(v)$ such that if \PO's budget is strictly greater than $\thresh(v)$, he wins the game from $v$, and if his budget is strictly less than $\thresh(v)$, \PT wins the game. Moreover, it was shown that the threshold function $\thresh$ is the {\em unique} function that satisfies the following property, which we call the {\em average property}. Suppose that the sum of budgets is $1$, and $t_i$ is \PLi's target, for $i \in \set{1,2}$. Then, $\thresh$ assigns a value in $[0,1]$ to each vertex such that at the ``end points'', we have $\thresh(t_1) = 0$ and $\thresh(t_2) = 1$, and the threshold at every other vertex is the average of two of its neighbors. 

Uniqueness implies that the problem of finding threshold budgets\footnote{Stated as a decision problem: given a game and a vertex $v$, decide whether $\thresh(v) \geq 0.5$.} is in \NP and co\NP. Moreover, an intriguing equivalence was observed between reachability continuous-bidding  games and a class of stochastic game~\cite{Con92} called {\em random-turn games}~\cite{PSSW09}. Intricate equivalences between {\em mean-payoff} continuous-bidding games and random-turn games have been shown in \cite{AHC19,AHI18,AHZ21,AJZ21} (see also~\cite{AH20}). 

Parity continuous-bidding games were studied in~\cite{AHC19}. The following key property was identified in games played on strongly-connected graphs. With every positive initial budget, a player can force visiting all vertices in the graph infinitely often. Consider a strongly-connected parity continuous-bidding game $\G$. It follows that if the maximal parity index in $\G$ is odd, then \PO wins with any positive initial budget, i.e., the thresholds in $\G$ are all $0$. Dually, if the maximal parity index in $\G$ is even, then the thresholds are all $1$. This property gives rise to a simple reduction from parity bidding games to double-reachability bidding  games: roughly, a player's goal is to reach a bottom strongly-connected component in which he can win with any positive initial budget.

\subparagraph*{Discrete-bidding games.}
This paper studies {\em discrete-bidding games}, which are similar to continuous-bidding games except that the granularity of the bids is restricted: the sum of the budgets in the game is fixed to $k \in \Nat$ and bids are restricted to be integers. A key difference between continuous- and discrete-bidding games is bidding ties, which now need to be handled explicitly. We focus on the tie-breaking mechanism that was defined in~\cite{DP10}: one of the players has the {\em advantage} and when a tie occurs, the player with the advantage chooses between (1)~use the advantage to win the bidding and pass it to the other player, or (2)~keep the advantage and let the other player win. Other tie-breaking mechanisms and the properties that they lead to were studied in~\cite{AAH21}. For example, the tie-breaking mechanism ``alternate tie breaking'' in which the players alternate turns in winning ties, is not determined, i.e., there is a game with an initial position such that neither player has a (pure) winning strategy. On the other hand, it was shown that any tie-breaking mechanism that breaks ties without considering the history of ties, admits determinacy.

The motivation to study discrete-bidding games is practical: in most applications, the assumption that bids can have arbitrary granularity is unrealistic. We point out that the results in continuous-bidding games, particularly those on infinite-duration games, do in fact develop strategies that bid arbitrarily small bids. It is highly questionable whether such strategies are applicable in practical applications. 

Bidding games model ongoing and stateful auctions. We list examples of domains in which such auctions arise. An immediate example is auctions for online advertisements~\cite{Mut09}. 
Bidding games were applied in~\cite{AMS24} as a scheduling mechanism in a ``decoupled''~synthesis procedure: given an objective of the form $\psi_1 \wedge \psi_2$, the idea is to find, independently, two bidding strategies $f_1$ for $\psi_1$ and $f_2$ for $\psi_2$  and an initial budget allocation such that the strategies guarantee that the outcome of playing $f_1$ against $f_2$ satisfies $\psi_1 \wedge \psi_2$. For example, the strongest guarantees are obtained when each strategy $f_i$, for $i =1,2$, is winning for $\psi_i$, i.e., it guarantees $\psi_i$ even against an adversary. 
Bidding as a mechanism for scheduling arises in {\em blockchain} technology, where {\em miners} accept transaction fees, which can be thought of as bids, and prioritize transactions based on them. Verification against attacks is a well-studied problem~\cite{CGV18,ABC16}. Attacks based on manipulations of these fees are hard to detect, can cause significant losses, and thus call for verification of the protocols~\cite{CGV18,ABC16}. 
Another example is applying bidding games as a mechanism for fair allocation of resources. Non-zero-sum Richman-bidding games were studied and applied to resource allocation in~\cite{MKT18} and {\em poorman}-bidding games (in which winning bids are paid to the ``bank'') were applied in~\cite{BEF21}. Poorman discrete-bidding were studied in~\cite{AM+23}.
In addition, researchers have studied training of agents that accept ``advice'' from a ``teacher'', where the advice is equipped with a ``bid'' that represents its importance \cite{AK+16}.
Finally, recreation bidding games have been studied, e.g., bidding chess~\cite{BP09}, as well as combinatorial games that apply bidding instead of alternating turns~\cite{RLP21}. 

We reiterate that practical applications of bidding games require some granularity on the bids. At the same time, we seek a high granularity to enable flexibility. A high granularity translates to  choosing a large sum of budgets $k$.

\subparagraph*{Previous results.}
For reachability objectives, the theory of continuous bidding games was largely adapted to discrete-bidding in~\cite{DP10}: threshold budgets were shown to exist and satisfy a discrete version of the average property and  winning strategies are derived from the threshold budgets. However, the only known algorithm to compute thresholds is a value-iteration algorithm whose worst-case running time is exponential when $k$ is given in binary.

For parity discrete-bidding games, there were large gaps in our understanding. In~\cite{AAH21}, {\em determinacy} was shown, namely from each configuration of the game, one of the players has a (pure) winning strategy. Determinacy is achieved by showing that the game satisfies a local property, which implies ``global'' determinacy. Using the observation that an additional budget will not harm a player, we obtain existence of thresholds. Importantly, this technique does not show that the thresholds satisfy the average property, and it was left open whether they indeed satisfy the average property. In terms of complexity, the algorithm to decide the winner from a configuration is naive: construct and solve the explicit concurrent game that corresponds to a bidding game. The running time of the algorithm is exponential when $k$ is given in binary. Another disadvantage of the algorithm is that the strategies that it produces use exponential memory and do not connect between bids and thresholds, as is done in reachability discrete-bidding games. 
To make matters worse, it was observed that unlike the properties of thresholds in reachability discrete-bidding games, which are conceptually similar to those in continuous-bidding, thresholds in parity discrete- and continuous-bidding games are inherently different: a strongly-connected discrete-bidding game $\G$ was shown in which a player cannot force visiting a vertex $v$ infinitely often even if he is initially allocated all of the budget. That is, when $\G$ is a B\"uchi game and $v$ is the only accepting vertex, then under continuous-bidding semantics, the threshold are $0$ whereas under discrete-bidding, the thresholds are $k+1$ (meaning that even a budget of $k$ does not suffice for winning).

\subparagraph*{Our results.}
We develop two complementary algorithms for computing threshold budgets in parity discrete-bidding games. 
Our first algorithm is a fixed-point algorithm. It repeatedly solves (i.e., finds threshold budgets) in {\em frugal-reachability} bidding games, which is an objective that we introduce in which on top of a reachability objective, in order to win, a player must reach its target with a sufficient budget. 
Our algorithm is inspired by algorithms to solve turn-based games such as Zielonka's~\cite{Zie98} and Kupferman and Vardi's~\cite{KV98} algorithms, whereas continuous-bidding games reduce to stochastic games. 
Recently, the fixed-point algorithm was adapted to bidding games {\em with charging}~\cite{AGHM24}.
This algorithm shows, for the first time, that threshold budgets in parity discrete-bidding games satisfy the average property. Moreover, the strategies that it produces are derived from the thresholds, as in reachability discrete-bidding games. 
On the downside, the algorithm runs in exponential time when $k$ is given in binary.

Second, we show that the problem of finding threshold budgets in parity discrete-bidding games\footnote{Formally, given a discrete-bidding game $\G$, a vertex $v$, and a threshold $\ell$, decide whether $\thresh(v) \geq \ell$.} 
is in NP and coNP. The bound applies also to reachability discrete-bidding games for which only an exponential-time algorithm was known. 
We briefly describe the idea of our proof. A first attempt to find thresholds is to guess thresholds (this is possible since the budgets are discrete) and verify that the guess satisfies the average property (recall that in continuous-bidding games, functions that satisfy the average property are unique). 
We show, however, that functions that satisfy the discrete average property are not unique. That is, even if a function satisfies the average property, it might not represent the thresholds in the game. We point out that this observation holds already in reachability discrete-bidding games, and to the best of our knowledge, was never made before. 
We overcome this challenge as follows. Our algorithm first guesses a function, checks whether it satisfies the average property, then verifies that it coincides with the thresholds. This last step is done via a reduction to turn-based parity games and is based on the structure of the thresholds and strategies that our first algorithm establishes. 
Another advantage of this algorithm is that it outputs a strategy that can be implemented using linear memory (previously, only construction of exponential-size strategies was known).

The new version of this paper corrects typos that previously appeared in crucial places and addresses the corner case of losing vertices in Lem.~\ref{lemma:TBparity-iff-thresh} and Lem.~\ref{lemma:thresholdpassesthetest}.

\section{Preliminaries}

\subsection{Concurrent games}
We define the formal semantics of bidding games via two-player {\em concurrent games}~\cite{AHK02}. Intuitively, a concurrent game proceeds as follows. A token is placed on a vertex of a graph. In each turn, both players concurrently select actions, and their joint actions determine the next position of the token. The outcome of a game is an infinite path. A game is accompanied by an objective, which specifies which plays are winning for \PO. 
In this paper, we will consider {\em reachability} and {\em parity} objectives. 

Formally, a concurrent game is played on an {\em arena} $\zug{A, Q, \lambda, \delta}$, where $A$ is a finite non-empty set of actions, $Q$ is a finite non-empty set of states (in order to differentiate, we use ``states'' or ``configurations'' in concurrent games and ``vertices'' in bidding games), the function $\lambda: Q \times \set{1,2} \rightarrow 2^A \setminus \{\emptyset\}$ specifies the allowed actions for \PLi in vertex $v$, and the transition function is $\delta: Q \times A \times A \rightarrow Q$. Suppose that the token is placed on a state $q \in Q$ and, for $i \in \set{1,2}$, \PLi chooses action $a_i \in \lambda(q, i)$. Then, the token moves to $\delta(q, a_1, a_2)$. For $q,q' \in Q$, we call $q'$ a {\em neighbor} of $q$ if there is a pair of actions $\zug{a_1,a_2} \in \lambda(q, 1) \times \lambda(q, 2)$ with $q' = \delta(q, a_1, a_2)$. We denote the neighbors of $q$ by $N(q) \subseteq Q$.

A (finite) {\em history} is a sequence  $\zug{q_0,a^1_0, a^2_0},\ldots, \zug{q_{n-1},a^1_{n-1},a^2_{n-1}}, q_n \in (Q \times A \times A)^*\cdot Q$ such that, for each $0 \leq i <n$, we have $q_{i+1} = \delta(q_i, a^1_i, a^2_i)$. 
A {\em strategy} is a ``recipe'' for playing the game. Formally it is a function $\sigma: (Q \times A \times A)^*\cdot Q \rightarrow A$. We restrict attention to {\em legal} strategies; namely, strategies that for each history $\pi \in (Q \times A \times A)^*\cdot Q$ that ends in $q \in Q$, choose an action in $\lambda(q, i)$, for $i \in \set{1,2}$. A {\em memoryless} strategy is a strategy that, for every state $q \in Q$, assigns the same action to every history that ends in $q$. 

Two strategies $\sigma_1$ and $\sigma_2$ for the two players and an initial state $q_0$, give rise to a unique {\em play}, denoted $\play(q_0, \sigma_1, \sigma_2)$, which is a sequence in $(Q \times A \times A)^\omega$ and is defined inductively as follows. The first element of $\play(q_0, \sigma_1, \sigma_2)$ is $q_0$. Suppose that the prefix of length $j \geq 1$ of \(\play(q_0, \sigma_1, \sigma_2)\) is defined to be $\pi^j \cdot q_j$, where $\pi^j \in (Q \times A \times A)^*$. Then, at turn $j$, for $i \in \set{1,2}$, \PLi takes action $a^j_i = \sigma_i(\pi^j \cdot q_j)$, the next state is $q^{j+1} = \delta(q_j, a^j_1, a^j_2)$, and we define $\pi^{j+1} = \pi^j \cdot \zug{v_j, a^j_1, a^j_2} \cdot q_{j+1}$. The {\em path} that corresponds to $\play(q_0, \sigma_1, \sigma_2)$ is $q_0,q_1,\ldots \in Q^\omega$. 

For $i \in \set{1,2}$, we say that \PLi ~{\em controls} a state $q \in Q$ if, intuitively, the next state is determined solely according to their chosen action. Formally, $q$ is controlled by \PO if for every action $a_1 \in A$, there is a state $q'$ such that no matter which action $a_2 \in A$ \PT takes, we have $q' = \delta(q, a_1, a_2)$, and the definition is dual for \PT. {\em Turn-based games} are a special case of concurrent games in which all states are controlled by one of the players.  Note that a concurrent game that is not turn-based might still contain some vertices that are controlled by one of the players.

\subsection{Bidding games}
A discrete-bidding game is played on an arena $G = \zug{V, E, k}$, where $V$ is a set of vertices, $E \subseteq V \times V$ is a set of directed edges, and $k \in \Nat$ is the sum of the players' budgets. For a vertex $v \in V$, we slightly abuse notation and use $N(v)$ to denote the neighbors of $v$ in $G$, namely $N(v) = \set{u: E(v,u)}$. We will consider decision problems in which $G$ is given as input. We then assume that $k$ is encoded in binary, thus the size of $G$ is $O(|V| + |E| + \log(k))$.

Intuitively, in each turn, both players simultaneously choose a bid that does not exceed their available budgets. The higher bidder moves the token and pays the other player. Note that the sum of budgets is constant throughout the game. 
Tie-breaking needs to be handled explicitly in discrete-bidding games as it can affect the properties of the game \cite{AAH21}. In this paper, we focus on {\em advantage-based} tie-breaking mechanism \cite{DP10}: exactly one of the players holds the {\em advantage} at a turn, and when a tie occurs, the player with the advantage chooses between (1) win the bidding and pass the advantage to the other player, or (2) let the other player win the bidding and keep the advantage.
We describe the semantics of bidding games formally below.

We will describe the formal semantics of a bidding game by constructing the explicit concurrent game that it corresponds to. We introduce the required notation. Following \cite{DP10}, we denote the advantage with $*$. Let $\Nat$ denote the non-negative integers, $\Nat^*$ the set $\set{0, 0^*, 1, 1^*, 2, 2^*, \ldots}$, and $[k]$ the set $\set{0,0^*,\ldots, k, k^*}$. We define an order $<$ on $\Nat^*$ by $0 < 0^* < 1 < 1^* < \ldots$. Let $m \in \Nat^*$. When saying that \PO has a budget of $m^* \in [k]$, we mean that \PO has the advantage, and implicitly, we mean that \PT's budget is $k-m$, and she does not have the advantage. 
We use $|m|$ to denote the integer part of $m$, i.e., if \(m = x^*\) for some \(x \in \Nat\), we denote \(|m| = x\).
Specifically, for \(m \in \Nat\), we have \(|m| = m\).  
We define operators $\oplus$ and $\ominus$ over $\Nat^*$.
Intuitively, we use \(\oplus\) as follows: suppose that \PO's budget is $m^*$ and \PT wins a bidding with a bid of $b_2$, then \PO's budget is updated to $m^* \oplus b_2$. Similarly, for $\ell \leq m$, a bid of $b_1 = \ell^*$ means that \PO will use the advantage if a tie occurs and $b_1 = \ell$ means that he will not use it. Upon winning the bidding, his budget is updated to $m^* \ominus b_1$.

  \begin{definition}[\(\oplus\) and \(\ominus\) operators]
    For \(x, y \in \bbN\), define $x^* \oplus y = x \oplus y^* = (x + y)^*$, $x \oplus y = x+y$. For $x,y \in \Nat$, define \(x \ominus y = x-y\), \(x^* \ominus y = (x-y)^*\), and in particular $x^* \ominus y^* = x-y$.
	For notational consistency, for \(x, y \in \Nat\), we define \(x^* \oplus y^* = (x+y+1)\), and \(x \ominus y^* = (x - y - 1)^*\). 
	Recall that $\ominus$ is intuitively used to deduct the winning bid from the winner's budget and \(\oplus\) is intuitively used to add the winning bid to the losing player's budget,
	hence the latter two cases do not follow this intuitive meaning.
\end{definition}

Next, we highlight two special cases, which are used frequently throughout the paper.

\begin{definition}
[Successor and predecessor]
For $B \in \Nat^*$, we denote by $\succb{B}$ and $\predb{B}$ respectively the {\em successor} and {\em predecessor} of $B$ in $\Nat^*$ according to $<$, defined as $\succb{B} = \min \set{x > B}$ and $\predb{B} = \max \set{x < B}$.
We note that this notation is convenient since it applies both for budgets that include the advantage and those that do not. When the status of the advantage is known we use the following notation. When \(B \in \Nat\) does not include the advantage, we use \(B^*\) as shorthand for \(\succb{B}\). 
When \(B = x^*\) for some \(x \in \Nat\), we use \(|B|+1 = x+ 1\) as shorthand for~\(\succb{B}\). 
\end{definition}

\subsection{Bidding games as concurrent games}
\label{sec:bidding2conc}
Consider an arena $\zug{V, E, k}$ of a bidding game. 
The corresponding {\em configurations} are $\C = \set{\zug{v, B} \in V \times ([k] \cup \{k+1\})}$, where a configuration $c = \zug{v, B} \in \C$ means that the token is placed on vertex \(v \in V\) and \PO's budget is \(B\). Implicitly, \PT's budget is \(k^* \ominus B\). 
Note that, vertices of the form \(\zug{v, k+1} \in \C\) are symbolic, namely they do not represent configurations of the game since \PO's budget cannot exceed $k$. 
The arena of the explicit concurrent game is $\zug{A, \C, \lambda, \delta}$, where $A = [k] \times V$, and we define the allowed actions in each configuration and transitions next. 
An action \(\zug{b, v} \in A\) means that the player bids \(b\) and proceeds to \(v\) upon winning the bidding. 
We require the player with the advantage to decide prior to the bidding whether they will use the advantage or not. 
Thus, when \PO's budget is $B^*$, \PO's legal bids are $[B]$ and \PT's legal bids are $\set{0,\ldots, k- B}$, and when \PO's budget is $B$, \PO's legal bids are $\set{0,1,\ldots, B}$ and \PT's legal bids are $[k \ominus B]$.
Next, we describe the transitions.
Suppose that the token is placed on a configuration $c = \zug{v, B}$ and \PLi chooses action $\zug{b_i, u_i}$, for $i \in \set{1,2}$. If $b_1 > b_2$, \PO wins the bidding and the game proceeds to $\zug{u_1, B_1 \ominus b_1}$. The definition for $b_2 > b_1$ is dual. 
The remaining case is a tie, i.e., $b_1 = b_2$. Since only one of the players has the advantage, a tie can occur only when the player who has the advantage does not use it. Suppose that $c = \zug{v, B^*}$, i.e., \PO has the advantage, and the definition when \PT has the advantage is dual. 
\PT wins the bidding, \PO keeps the advantage, and we proceed to $\zug{u_2, B^* \oplus b_2}$. 
Note that the size of the arena is $O(|V|  \times k)$, which is exponential in the size of $\G$ since $k$ is given in binary.

Consider two strategies $f$ and $g$ and an initial configuration $c = \zug{v, B}$
We sometimes abuse notation and treat $\play(v, f, g) = \zug{v_0, B_0}, \zug{v_1, B_1},\ldots$ as the infinite path $v_0,v_1,\ldots$ in the bidding game.

\subsection{Objectives and threshold budgets}
\label{sec:prelim-obj}
A bidding game is $\G = \zug{V, E, k, \Obj}$, where $\zug{V, E, k}$ is an arena and $\Obj \subseteq V^\omega$ is an {\em objective}, which specifies the infinite paths that are winning for \PO. 

We introduce notations on paths before defining the objectives that we consider. Consider a path $\pi = v_0,v_1,\ldots$ and consider a subset of vertices $A \subseteq V$. We say that $\pi$ {\em visits} $A$ if there is $j \geq 0$ such that $v_j \in A$. We denote by $inf(\pi) \subseteq V$, the set of vertices that $\pi$ visits infinitely often. 
We say that $\pi$ {\em enters} $A$ at time $j \geq 1$ if $v_j \in A$ and $v_{j-1} \notin A$, and it is {\em exited} at time $j$ if $v_j \notin A$ and $v_{j-1} \in A$. 

We consider the following two canonical objectives:
\begin{itemize}
\item {\bf Reachability:} A reachability bidding game is $\zug{V, E, k, S}$, where $S \subseteq V$ is a set of sinks. \PO, the reachability player, wins an infinite play $\pi$ iff it visits $S$, and we then say that $\pi$ {\em ends} in $S$. {\em Safety} objectives are dual to reachability objectives; the safety player wins a play iff it never visits $S$.  
\item {\bf Parity:}  A parity bidding game is $\zug{V, E, k, p}$, where $p: V \rightarrow \set{1,\ldots, d}$ assigns to each vertex a {\em parity index}, for $d \in \Nat$. A play $\tau$ is winning for \PO iff $\max_{v \in inf(\tau)} p(v)$ is odd. The special case in which $p$ assigns parities in $\set{2,3}$ is called {\em B\"uchi} objective;  \PO wins a play iff it visits the set $\set{v \in V: p(v) = 3}$ infinitely often. 
\end{itemize}

We introduce {\em frugal} objectives in bidding games in which, roughly, \PO wins by reaching a target with a sufficient budget.
\begin{definition}
[Frugal objectives]
\mbox{}\newline
\begin{itemize}
\item A {\em frugal-reachability} bidding game is $\zug{V, E, k, S, \fr}$, where $V$, $E$, and $k$ are as in bidding games, $S \subseteq V$ is a set of target vertices, and $\fr:S \rightarrow [k]$ assigns a {\em frugal-target budget} to each target. Consider a play $\pi$. 
 \PO wins $\pi$ iff $\pi$ reaches $S$, i.e., $\pi$ ends in a configuration \(\zug{s, B}\) with \(s \in S\), and \PO's budget at $S$ exceeds the frugal-target budget, i.e., $\pi$ ends in \(\zug{s, B}\) with $B \geq \fr(s)$. Note that a reachability bidding game is a special case of a frugal-reachability bidding game in which $\fr \nequiv 0$. 
\item The {\em frugal-safety} objective is dual to frugal-reachability. We describe the winning condition explicitly. A frugal-safety bidding game is $\zug{V, E, k, S, \fr}$, where $V$, $E$, and $k$ are as in bidding games, $S \subseteq V$ is a set of sinks, and $\fr:S \rightarrow [k]$ assigns a frugal-target budget to each sink. \PO, the safety player, wins a play $\pi$ if: (1) $\pi$ never reaches $S$, or (2) $\pi$ reaches a configuration $\zug{s, B}$ with $s \in S$ and $B \geq \fr(s)$. Note that a safety bidding game is a special case of a frugal-safety bidding game in which $\fr \nequiv k+1$. 
\item A {\em frugal-parity} bidding game is $\zug{V, E, k, p, S, \fr}$, where $p: (V \setminus S) \rightarrow \set{0,\ldots, d}$ and the other components are as in the above. \PO wins a play $\pi$ if (1) $\pi$ does not reach $S$ and satisfies the parity objective, or (2) $\pi$ satisfies a frugal-reachability objective: it ends in a configuration $\zug{s, B}$ with $s \in S$ and $B \geq \fr(s)$.
\end{itemize}
\end{definition}

\begin{remark}
We point out that an $\omega$-regular objective (e.g., reachability and parity) is a collection of paths in the graph, i.e., a subset of $V^\omega$. On the other hand, a frugal objective is a subset of configurations, i.e., a subset of $(V \times [k])^\omega$. We do not know of a reduction from a frugal objective game to a non-frugal objective game that preserves the thresholds. At the same time, as we show in the next section, it is not hard to extend algorithms for reachability bidding games to frugal-reachability games. Moreover, we are not aware of applications of frugal objectives, and define them as a means to solve parity bidding games. 
\end{remark}

Next, we define winning strategies.
\begin{definition}[Winning strategies]
Consider a configuration $c = \zug{v, B}$ and an objective~\(\Obj\).
\begin{itemize}
\item A \PO strategy $f$ is winning from \(c\) if for every strategy $g$, $\play(c, f, g)$ satisfies \(\Obj\). 
\item A \PT strategy \(g\) is winning from \(c\) if for every strategy \(f\), \(\play(c, f, g)\) does not satisfy~\(\Obj\). 
\end{itemize}
A player wins from \(c\) if they have a winning strategy from \(c\).
\end{definition}

The central quantity in bidding games is the {\em threshold budget} at a vertex, which is the necessary and sufficient initial budget at that vertex for \PO to guarantee winning the game. It is formally defined as follows.

\begin{definition}[Threshold budgets] Consider a bidding game $\G$. The {\em threshold budget} at a vertex $v$ in $\G$, denoted $\thresh_\G(v)$, is such that 
\begin{itemize}
\item \PO wins from every configuration $\zug{v, B}$ with $B \geq \thresh_\G(v)$, and 
\item \PT wins from every configuration $\zug{v, B}$ with $B < \thresh_\G(v)$.
\end{itemize}
We refer to the function $\thresh_\G$ as the {\em threshold budgets}. 
\end{definition}

\begin{remark}[Thresholds in losing vertices]
\label{rem:losing-vertices}
Note that a game might contain vertices that are losing for \PO with every initial budget. For example, in a reachability game, a sink with no path to the target is always losing for \PO. Following~\cite{DP10}, in such a losing vertex $v$, we set the threshold to $\thresh(v) = k+1$. Since the highest possible budget for a player in the game is $k^*$, setting $\thresh(v) = k+1$ can intuitively be understood as \PO requires more budget for winning than he can possibly obtain.  
\end{remark}

\begin{remark}
We point out that existence of threshold budgets is not trivial. 
Indeed, existence of threshold budgets implies {\em determinacy}: from each configuration, there is a player who has a winning strategy. 
By observing that an additional budget will not harm a player in a bidding game, we obtain that determinacy implies existence of thresholds.
Recall that bidding games are succinctly-represented concurrent games, and concurrent games are often not determined, for example, the simple concurrent game ``matching pennies'' is not determined since neither player can win the game. 
Typically, the vertices of a concurrent game can be partitioned into {\em surely winning} vertices from which \PO has a winning strategy, {\em surely losing} from which \PT has a winning strategy, and in the rest, neither player has a winning strategy. An optimal strategy from a vertex in the last set is {\em mixed}; it assigns a probability distribution over actions. 
Interestingly, in bidding games, all vertices are either surely winning or surely losing. 
Previously, existence of thresholds in reachability games with advantage-based tie-breaking was shown in~\cite{DP10} by identifying the structure of the thresholds. An alternative technique was developed in~\cite{AAH21} in which bidding games were shown to have a local property (called ``local determinacy'') that implies ``global'' determinacy, and was used to show that Muller bidding games with various tie-breaking mechanisms are determined. Other subclasses of concurrent games (beyond bidding games) were shown to have a local property that implies determinacy in~\cite{BBR21}.
\end{remark}

\section{Frugal-Reachability Discrete-Bidding Games}
The study in \cite{DP10} focuses primarily on reachability discrete-bidding games played on DAGs. 
We revisit their results, provide explicit and elaborate proofs for games played on general graphs, and extend the results to frugal-reachability games. 
Specifically, Theorem~\ref{thm:non-unique} points to an issue in bidding games played on general graphs that was not explicitly addressed in~\cite{DP10}.

\subsection{Background: reachability continuous-bidding games}
Many of the techniques used in reachability discrete-bidding games are adaptations of techniques developed for reachability continuous-bidding games~\cite{LLPU96,LLPSU99}. 
In order to develop intuition and ease presentation of discrete-bidding games, in this section, we illustrate the ideas and techniques of continuous-bidding games.

Recall that in continuous-bidding games there is no restriction on the granularity of bids, i.e., bids can be arbitrarily small. Throughout this section we assume that the sum of the players' budgets is $1$. Note that since winning bids are paid to the opponent, the sum of budgets stays constant throughout the game. 

\begin{definition}[Continuous threshold budgets]\label{def:cont-thresh}
The {\em continuous threshold budget} at a vertex $v$ is a budget $\thresh(v) \in [0,1]$ such that for every $\epsilon > 0$: 
\begin{itemize}
\item if \PO's budget is $\thresh(v) + \epsilon$, he wins the game from $v$, and 
\item if \PO's budget is $\thresh(v) - \epsilon$, \PT wins the game from $v$.
\end{itemize}
\end{definition}

\begin{remark}[Losing vertices in continuous-bidding games]
Recall that in discrete bidding, the threshold budget in a losing vertex is $k+1$, which is higher than the highest budget a player can obtain in a discrete bidding game (Remark~\ref{rem:losing-vertices}). 
The treatment of losing vertices in continuous bidding is similar, though implicit in Def.~\ref{def:cont-thresh}: the continuous threshold budget in a losing vertex is $1$, thus by Def.~\ref{def:cont-thresh}, \PO wins if his budget is $1+\epsilon$, which is never the case since it is higher than the total budget. 
\end{remark}

\begin{remark}
We point out that the issue of tie breaking is avoided in continuous-bidding games by considering initial budgets that differ from the threshold. 
That is, the guarantee is that if a player's budget is strictly above the threshold, he wins the game no matter which tie-breaking mechanism is used, e.g., even if the opponent wins all bidding ties. 
\end{remark}

A {\em double-reachability} continuous-bidding game is $\zug{V, E, t_1, t_2}$, where for $i \in \set{1,2}$, the vertex $t_i$ is the target of \PLi and every vertex $v \neq t_1,t_2$ has a path to both. The game ends once one of the targets is reached, and the player whose target is reached is the winner. The careful reader might notice that the definition does not define a winner when no target is reached. We will show below that this case is avoided.

\begin{definition}[Continuous average property]
\label{def:continuous-average}
Consider a double-reachability continuous-bidding game $\G = \zug{V, E, t_1, t_2}$ and a function $T: V \rightarrow [0,1]$. For $v \in V$, denote $v^+_T := \arg\max_{u \in N(v)} T(u)$ and $v^-_T := \arg\min_{u \in N(v)} T(v)$. We say that $T$ has the {\em continuous average property} if for every vertex $v \in V$: 
\[T(v) = \begin{cases}
1 & \text{ if } v = t_2 \\
0 & \text{ if } v = t_1 \\
\frac{T(v^-_T) + T(v^+_T)}{2} & \text{ otherwise}
\end{cases}
\]
when the function \(T\) is clear from the context, we simply refer \(v^+_T\) and \(v^-_T\) as \(v^+\) and \(v^-\) respectively. 
Note that, there could be more than one vertex \(u\) (similarly, \(w\)) such that \(T(u) = T(v^-)\) (respectively, \(T(w) = T(v^+)\)), but for the sake of convenience, we collectively denote any of them as \(v^-\) and \(v^+\) respectively. 
\end{definition}

The next theorem presents the main results on reachability continuous-bidding games: a function that satisfies the continuous average property is unique, and it coincides with the continuous threshold budgets. We illustrate the proof techniques, in particular how to construct a winning bidding strategy given the thresholds in the game.

\begin{restatable}{theorem}{contreach}
\label{thm:cont-reach}
\cite{LLPU96,LLPSU99}
Consider a double-reachability continuous-bidding game $\zug{V, E, t_1, t_2}$. Continuous threshold budgets exist, and the threshold budgets $\thresh: V \rightarrow [0,1]$ is the unique function that has the continuous average property. Moreover, the problem of deciding whether $\thresh(v) \leq 0.5$ is in NP and coNP.
\end{restatable}
\begin{proof}
(SKETCH)
Let $T$ be a function that satisfies the continuous average property, where we omit the proof of existence of such a function. We prove that for every vertex $v$, the continuous threshold budget at $v$ is $T(v)$. Uniqueness follows immediately. 

The complexity bound is obtained by guessing, for every vertex $v$, two neighbors $v^-$ and $v^+$, and constructing and solving a linear program based on the constraints in Def.~\ref{def:continuous-average} (for more details, see~\cite{AHC19}, which shows a reduction to stochastic games~\cite{Con92}). 

\subparagraph*{From a function with the continuous-average property to a strategy.}
Suppose that \PO's budget at $v$ is $T(v) + \eps$, for $\eps > 0$. We describe a winning \PO strategy. 
Recall that $v^+, v^- \in N(v)$ are respectively the neighbors of $v$ that attain the maximal and minimal values according to $T$. Let 
\[b(v) := \frac{T(v^+) - T(v^-)}{2}.\]
The key observation is that $T(v) + b(v) = T(v^+)$ and $T(v)-b(v) = T(v^-)$.

Consider the following \PO strategy. At vertex $v \in V$, bid $b(v)$ and proceed to $v^-$ upon winning. We show that the strategy maintains the following invariant:

\noindent{\bf Invariant:}
When the game reaches a configuration $\zug{u, B}$, then $B > T(u)$. 

We list two implications of the invariant. First, it implies that the strategy is legal, namely \PO's budget at $v$ suffices to bid $b(v)$. 
Second, it implies that \PO does not lose, namely no matter how \PT plays, the game will not reach $t_2$. Indeed, assume towards contradiction that $t_2$ is reached. Then, the invariant implies that \PO's budget is strictly greater than $1$, which violates the assumption that the sum of budgets is $1$.

Note that ``not losing'' does not suffice for winning, namely that \PO forces the game to $t_1$. These details, however, are not relevant to this paper. For completeness, we describe the rough idea. 
Suppose that the game reaches configuration $\zug{u, B}$. The invariant implies $B > T(u)$. We call $B - T(u)$ \PO's {\em spare change}. 
The idea is to choose \PO's bids carefully in a way that ensures that as the game continues, his spare change strictly increases so that eventually his budget suffices to win $|V|$ times in a row. We point out that this idea can be extended to show that in a strongly-connected game, a player can force infinitely many visits to a vertex with any positive initial budget, which is at the core of solving parity continuous-bidding games. 

We prove that \PO's strategy maintains the invariant against any \PT strategy. 
Note that the invariant holds initially. Suppose that the game reaches configuration $\zug{u, B}$ with $B > T(u)$. We claim that the invariant is maintained in the next turn. 
Indeed, if \PO wins the bidding, the next configuration is $\zug{v^-, B - b(v)}$, and the claim follows from $T(v)-b(v) = T(v^-)$. 
If \PT wins the bidding, she bids at least $b(v)$, thus \PO's updated budget is at least $B + b(v)$, and the worst that \PT can do for \PO is to move to $v^+$. The claim follows from $T(v) + b(v) = T(v^+)$. 

\subparagraph*{Reasoning about the flipped game.}
Finally, we show that \PT wins when \PO's budget is $T(v) - \eps$. We intuitively ``flip'' the game and associate \PO with \PT. 
More formally, let $\G'$ be the same as $\G$ except that \PO's goal is to reach $t_2$ and \PT's goal is to reach $t_1$. 
For every $u \in V$, define $T'$ by $T'(u) = 1-T(u)$. A key observation is that $T'$ satisfies the continuous average property in $\G'$. 
In particular, note that $T'(t_1) = 1$ and $T'(t_2) = 0$. 
Now, in order to win from $v$ in $\G$ when \PO's budget is $T(v) - \varepsilon$, \PT follows a winning \PO strategy in $\G'$ with an initial budget of $1-T(v) + \varepsilon$. 
\end{proof}

\subsection{Frugal-Reachability discrete-bidding games}

We turn to study discrete-bidding games. 

\subsubsection{The discrete average property}
\label{sec:average-based-strat}

In this section, we adapt the definition of the continuous average property (Def.~\ref{def:continuous-average}) to the discrete setting and analyze its properties. 

\begin{definition}[Average property]
\label{def:average}
Consider a frugal-reachability discrete-bidding game $\G = \zug{V, E, k, S, \fr}$. 
We say that a function $T: V  \rightarrow [k] \cup \set{k+1}$ has the {\em average property} if for every $s \in S$, we have $T(s) = \fr(s)$, and for every $v \in V \setminus S$, 
\[
T(v) = \floor{\frac{|T(v^+)| + |T(v^-)|}{2}} +\eps
\text{ where }
\eps = 
\begin{cases}
0 &\text{if~} |T(v^+)| + |T(v^-)| \text{ is even and~} T(v^-) \in \Nat \\
1 &\text{if~} |T(v^+)| + |T(v^-)| \text{ is odd and~} T(v^-) \in \Nat^* \setminus \Nat \\
* &\text{otherwise}
\end{cases}
\]
where $v^+ := \arg\max_{u \in N(v)} T(u)$ and $v^- := \arg\min_{u \in N(v)} T(v)$
\end{definition}

The following theorem shows, somewhat surprisingly and for the first time, that functions that satisfy the discrete average property are not unique. That is, there are functions satisfying the discrete average property but not coinciding with the threshold budgets. 
This is in stark contrast to continuous-bidding games in which there is a unique function that satisfies the average property.

\begin{theorem}
\label{thm:non-unique}
The reachability discrete-bidding game $\G_1$ that is depicted in Fig.~\ref{fig:avgbutnotthreshold} with target $t$ for \PO has more than one function that satisfies the average property. 
\end{theorem}
\begin{proof}
Assume a total budget of $k = 5$. We represent a function $T: V \rightarrow [k]$ as a vector $\zug{T(v_0), T(v_1), T(v_2), T(t)}$. It is not hard to verify that both $\zug{4,3^*,3,2}$ and $\zug{5,4^*,3^*,2}$ satisfy the average property. (The latter represents the threshold budgets). 
\end{proof}

\begin{figure}
  \begin{center}
  	\includegraphics[width=0.6\textwidth]{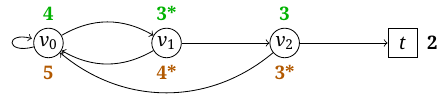}
      \end{center}
		

		
	\caption{A discrete-bidding reachability game with two functions that satisfy the average property.}\label{fig:avgbutnotthreshold}
\end{figure}

The following lemma intuitively shows that the ``complement'' of $T$, denoted below $T'$, satisfies the average property. For a vertex $v$, the value $T'(v)$ should be thought of as \PT's budget when \PO's budget is strictly less than $T(v)$. We will later show that since $T'$ satisfies the average property, \PT can win with a budget of $T'(v)$. A similar idea is used in continuous-bidding in the last point in the proof of Theorem~\ref{thm:cont-reach}. It follows that a function $T$ that satisfies the average property satisfies $T \leq \thresh_\G$. Indeed, on the one hand, if \PO's budget is less than $T(v)$, then \PT, the safety player, wins by avoiding \PO's targets. However, it could be the case that $T < \thresh_\G$; namely, \PO cannot win (force a visit to a target) with a budget of $T(v)$.

\begin{definition}[Complement of \(T\)]\label{def: ComplimentofT}
	Let \(\G = \zug{V, E, k, S, \fr}\) be a discrete-bidding game with a frugal objective. 
	Let \(T: V \rightarrow [k] \cup \{k+1\}\) be a function. 
	We define ``complement'' of \(T\), denoted by \(\invT: V \rightarrow [k] \cup \set{k+1}\) as: \(\invT(v) = (k+1) \ominus T(v)\) for all \(v \in V\). 
\end{definition}

\begin{observation}[Relation between \(T\) value and \(\invT\) value of a vertex]\label{obs: property}
	We make the following observations about such a function \(T\) and its complement \(\invT\), which can be directly derived using the definition of \(\ominus\), in the following.
	
	\begin{itemize}
		\item[\textbf{I.} ]
		For every $v \in V$, we have that $T(v)$ and $\invT(v)$ agree on which player has the advantage, formally \(T(v) \in \Nat\) iff \(\invT(v) \in \Nat\).  
		
		\item[\textbf{II.}] For $v \in V$, a neighbouring vertex with maximum \(T\)-value is the neighbouring vertex with minimum \(\invT\)-value, and vice-versa. 
		Notationally, \(v^+_T = v^-_{\invT}\) and \(v^-_T = v^+_{\invT}\)
		
		\item[\textbf{III}] For any \(v \in V\), if \(T(v) \in \Nat\) then \(\invT(v) = |\invT(v)| = (k+1) - |T(v)|\), otherwise \(|\invT(v)| = k - |T(v)|\). 		
	\end{itemize}
\end{observation}

\begin{restatable}{lemma}{flippedaverage}
	\label{lem:flippedaverage}
	Let $\G = \zug{V, E, k, S, \fr}$ be a discrete-bidding game with a frugal objective. Let \(T: V \rightarrow [k] \cup \{k+1\}\) be a function that satisfies the average property. 
	Then, the complement of \(T\), denoted by \(\invT: V \rightarrow [k] \cup \{k+1\}\),  also satisfies the average property. 

\end{restatable}

\begin{proof}
	Let us first fix the notation for \(\vminus\) and \(\vplus\) as to denote the neighbouring vertices of \(v\) with minimum and maximum \(T\)-value, respectively for this proof.  
	We need to show that \(\invT(v) = \floor{\frac{|\invT(v^-)|+ |\invT(v^+)|}{2}} + \eps'\), where \(\eps'\) comes from the definition of the average property.
	We already have \(T(v) = \floor{\frac{|T(v^+)|+ |T(v^-)|}{2}} + \eps\).
	Note that \(\invT(v^-) = \invT(v^+_{\invT})\) and \(\invT(v^+) = \invT(v^-_{\invT})\), as per Observation \ref{obs: property}.  
	
We proceed to prove the lemma. 

First, we consider the scenario when \(T(v) = 0\). 
This necessarily means \(\eps = 0\), \(T(\vminus) = 0\), \(|T(\vplus)| = 0\), and \(\invT(v) = (k+1) \ominus T(v) = k+1\). 
This implies \(\invT(\vplus{\invT}) = \invT(\vminus) = k+1\) and \(\invT(\vminus{\invT}) = \invT(\vplus) \in \{k^*, k+1\}\). 
If \(\invT(\vminus{\invT})\) is \(k^*\), we have \(|\invT(\vplus)| + |\invT(\vminus)| = 2k+1\) and \(\eps' = 1\). 
On the other hand, if \(\invT(\vminus{\invT}) = k+1\), then we have \(|\invT(\vplus)| + |\invT(\vminus)| = 2(k+1)\), and \(\eps' = 0\). 
Thus, we have \(\invT(v) = k+1 = \floor{\frac{\left|\invT\left(\vminus\right)\right| + \left|\invT\left(\vplus\right)\right|}{2}} + \eps'\)

\smallskip

Now, we assume that \(T(v) > 0\). We divide the analysis in four exhaustive cases in the following. 
Before delving into the case analysis, for better reading comprehension, we explain the structure that each of these analyses follows. 
We divide the analysis into four exhaustive cases according to the advantage statuses of \(T(\vplus)\) and \(T(\vminus)\), the first being where both of them are in \(\Nat\), in the second one, both of them are  in \(\Natstro\), and the final two cases are when one of them is in \(\Nat\) and the other one is in \(\Natstro\). 

Then, in each case, we look into the relation between  the parity of \(\sumT\) and \(\sumT{\invT}\), and what the \(\eps, \eps'\) values are. 
Once those are settled, we start with the corresponding expression of \(\floor{\frac{\sumT{\invT}}{2}} + \eps'\), and show that it is indeed \(\invT(v)\).

We now delve into the technical part of this analysis  below: 
	\begin{itemize}
		\item \(T(v^+), T(v^-) \in \Nat\).
		
		First, from Observation~\ref{obs: property}, we have \(\invT(\vminus), \invT(\vplus) \in \Nat\) as well. 
		
		Thus, it is enough to discuss the two cases corresponding to when the sum is even or odd, respectively. 
		Note that  \(|T(v^+)|+|T(v^-)|\) is even iff \(|\invT(v^-)|+|\invT(v^+)|\) is even, simply because in this case \(\invT(\vplus) = (k+1) - T(\vplus)\) and \(\invT(\vminus) = (k+1) - T(\vminus)\). 
		It follows that irrespective of whether \(|T(\vminus)| + |T(\vplus)|\) is odd or even, we have \(\eps' = \eps\). 
		
		\smallskip
		
		When \(\eps = \eps' = 0\) (i.e., the sum is even), we have
		\begin{align*}
			\frac{|\invT(v^+)|+|\invT(v^-)|}{2} &= \frac{(k+1 - |T(v^-)|) + (k+1 - |T(v^+)|)}{2}\\
			&=k+1 - \frac{\sumT}{2}\\
			&= (k+1) - T(v) = \invT(v) 
		\end{align*}
		
		When \(\eps = \eps' = 0^*\) (i.e., the sum is odd), 
		we have 
		\begin{align*}
			\left(\floor{\frac{|\invT(v^+)|+|\invT(v^-)|}{2}}\right)^* &= \left(\floor{k+1 - \frac{\sumT}{2}}\right)^*\\
			&= \left(k+1 - \floor{\frac{\sumT}{2}} - 1\right)^*\\
			&= (k+1) \ominus \left(\floor{\frac{\sumT}{2}}\right)^*\\
			&= (k+1) \ominus T(v) = \invT(v) 
		\end{align*}
	
		\smallskip
		
		\item \(T(v^+), T(v^-) \in \Natstro\).
		
		In this case too, \(|T(v^+)|+|T(v^-)|\) is even iff  \(|\invT(v^-)|+|\invT(v^+)|\) is even.
		Therefore, we have \(\eps' = \eps\).
		
		When \(\eps = \eps' = *\) (i.e., the sum is even and the minimum is in \(\Natstro\)), we have 
		\begin{align*}
			\left(\frac{|\invT(v^+)|+|\invT(v^-)|}{2}\right)^* &= \left(\frac{(k - |T(v^-)|)+(k - |T(v^+)|)}{2}\right)^*\\
			&= \left(k - \frac{|T(\vminus)| + |T(\vplus)|}{2}\right)^* \\
			&= \left((k+1) - \frac{|T(\vminus)| + |T(\vplus)|}{2} - 1\right)^*\\
			&= (k+1) \ominus \left(\frac{|T(\vminus)| + |T(\vplus)|}{2}\right)^*\\
			&= (k+1) \ominus T(v) = \invT(v) 
		\end{align*}
		
		When \(\eps = \eps' = 1\) (i.e., the sum is odd and the minimum is in \(\Natstro\)), we have 
		\begin{align*}
			\floor{\frac{|\invT(v^+)|+|\invT(v^-)|}{2}} + 1 &= \floor{\frac{(k - |T(\vplus)|) + (k - |T(\vminus)|)}{2}} + 1\\
			&= \floor{k - \frac{|T(\vminus)| + |T(\vplus)|}{2}} + 1\\
			&= k - \floor{\frac{|T(\vplus)| + |T(\vminus)|}{2}} - 1 + 1\\
			&= (k+1) - \left(\floor{\frac{|T(\vplus)| + |T(\vminus)|}{2}} + 1\right)\\
			&= (k+1) \ominus T(v) = \invT(v)
		\end{align*}
		
		\smallskip
		
		\item \(T(v^+) \in \Nat\) and \(T(v^-) \in \Natstro\)
		
		Note that, in this case, \(|T(v^+)|+|T(v^-)|\) is even iff  \(|\invT(v^-)|+|\invT(v^+)|\) is odd.
		
		We can also see that when \(|T(v^+)|+|T(v^-)|\) is even, we have \(\eps = \eps' = 0^*\). 
		Therefore, we have 
		\begin{align*}
			\left(\floor{\frac{|\invT(v^+)|+|\invT(v^-)|}{2}}\right)^* &= 
			\left(\floor{\frac{(k+1) - |T(\vplus)| + k - |T(\vminus)|}{2}}\right)^* \\
			&= \left(\floor{\frac{2k+1 - (|T(\vplus)|+|T(\vminus)|)}{2}}\right)^*\\
			&= \left(k - \frac{|T(\vplus)| + |T(\vminus)|}{2}\right)^*\\
			&= \left(k+1 - \frac{|T(\vplus)| + |T(\vminus)|}{2} - 1\right)^*\\
			&= (k+1) \ominus \left(\frac{|T(\vplus)| + |T(\vminus)|}{2}\right)^*\\
			&= (k+1) \ominus T(v) = \invT(v)
		\end{align*}
		
		On the other hand, when \(|T(\vplus)| + |T(\vminus)|\) is odd, we have \(\eps = 1, \eps' = 0\). 
		In this case, we proceed as follows:
		\begin{align*}
			\frac{|\invT(v^+)|+|\invT(v^-)|}{2} &= \frac{k - |T(\vplus)| + (k+1) - |T(\vminus)|}{2}\\
			&= \frac{2k+1 - (\absolut{T(\vplus)}+\absolut{T(\vminus)})}{2}\\
			&= k - \frac{\absolut{T(\vplus)} + \absolut{T(\vminus)} -1}{2}\\
			&= k - \floor{\frac{\absolut{T(\vplus)}+ \absolut{T(\vminus)}}{2}}\\
			&= (k+1) - \left(\floor{\frac{\absolut{T(\vplus)}+ \absolut{T(\vminus)}}{2}} + 1\right)\\
			&= (k+1) \ominus T(v) = \invT(v) 
		\end{align*}
		
		\smallskip
		
		\item \(T(v^+) \in \Nat^*\) and \(T(v^-) \in \Nat\)
		
		In this case, like the earlier one,  \(|T(v^+)|+|T(v^-)|\) is even iff  \(|\invT(v^-)|+|\invT(v^+)|\) is odd.
		But unlike the earlier case, here when \(\absolut{T(\vplus)} + \absolut{T(\vminus)}\) is even, we have \(\eps = 0\) and \(\eps' = 1\).
		On the other hand, when \(\sumT\) is odd, then we have \(\eps = \eps' = *\). 
		
		We first consider the case when \(\sumT\) is even, in the following:
		\begin{align*}
			\floor{\frac{\sumT{\invT}}{2}} + 1  &= \floor{\frac{k - \absolut{T(\vplus)} + (k+1) - \absolut{T(\vminus)}}{2}}\\
			&= \floor{\frac{2k+1 - (\sumT)}{2}}\\
			&= \floor{k - \frac{\sumT - 1}{2}}\\
			&= k - \floor{\frac{\sumT - 1}{2}} - 1\\
			&= k - \left(\frac{\sumT}{2} - 1\right) - 1\\
			&= (k+1) - \left(\frac{\sumT}{2} + 1\right)\\
			&= (k+1) \ominus T(v) = \invT(v) 
		\end{align*}
	
	On the other hand, when \(\sumT\) is odd, we have:
	\begin{align*}
		\left(\floor{\frac{\sumT{\invT}}{2}}\right)^* &= \left(\floor{k - \frac{\sumT - 1}{2}}\right)^*\\
		&= \left(k - \frac{\sumT - 1}{2}\right)^*\\
		&= \left(k - \floor{\frac{\sumT}{2}}\right)^*\\
		&= \left((k+1) - \floor{\frac{\sumT}{2}} - 1\right)^*\\
		&= (k+1) \ominus \left(\floor{\frac{\sumT}{2}}\right)^*\\
		&= (k+1) \ominus T(v) = \invT(v) 
	\end{align*}
	\end{itemize}

As we establish exhaustively that \(\floor{\frac{\sumT{\invT}}{2}}\) is indeed \(\invT(v)\), we conclude with the fact that  \(\invT\) satisfies the average property when \(T\) does. 
\end{proof}

    \subsubsection{From a function that satisfies the average to a bidding strategy}
\label{sec:def-f_T}
Throughout this section, fix a function $T: V \rightarrow [k] \cup \set{k+1}$ that satisfies the average property. We show how to construct a bidding strategy from $T$. We will use this construction also for parity games. 

\begin{definition}[Partial strategy]
A {\em partial strategy} based on \(T\) is a function \(f_T: \C \rightarrow [k] \times 2^V\) that chooses, at each configuration, a bid and a set of {\em allowed} vertices. 
\end{definition}
Note that $f_T$ is not a strategy since it does not assign a unique vertex to proceed to upon winning the bidding. 
 
\begin{definition}[$f'$ agrees with $f_T$]
Consider a strategy \(f'\) and a partial strategy $f_T: \C \rightarrow [k] \times 2^V$. Consider a configuration $c$. Let \(\zug{b, A} = f_T(c)\) and \(\zug{b', u'} = f'(c)\). We say that $f'$ agrees with $f_T$ at $c$ if $b=b'$ and $u' \in A$. We say that $f'$ agrees with \(f_T\) if $f'$ agrees with $f_T$ in all configurations. 
\end{definition}

We describe the intuition behind the construction of $f_T$. We construct $f_T$ so that every strategy $f'$ that agrees with $f_T$ maintains the following invariant. 
Suppose that the game starts from configuration $\zug{v, B}$ with $B \geq T(v)$. Then, against any opponent strategy, when the game reaches a configuration $\zug{u, B'}$, we have $B' \geq T(u)$. Since for every sink $s \in S$, we have $T(s) \geq \fr(s)$, the invariant implies that $f'$ guarantees that the frugal objective is not violated. Technically, the construction of $f_T$ is similar in spirit to the construction in continuous bidding (Theorem~\ref{thm:cont-reach}). There, a non-losing strategy maintains the invariant that \PO's budget exceeds $\thresh(v)$ by bidding $b = \frac{\thresh(v^+) - \thresh(v^-)}{2}$ at a vertex $v$, and proceeding to $v^-$ upon winning the bidding. Recall that the invariant is maintained since $\thresh(v) - b = \thresh(v^-)$ and $\thresh(v) + b = \thresh(v^+)$. 

We describe $f_T$ formally. Consider $v \in V$ and a budget $B \in [k]$ with $B \geq T(v)$. Let $v^+ = \arg \max_{u \in N(v)} T(u)$ and $v^- = \arg \min_{u \in N(v)} T(u)$. We define $f_T(\zug{v, B}) =~\zug{b^T(v,B), A(v)}$ as follows.
First, we define the allowed vertices
\begin{align}
\label{eq: allowedvertices}
A(v) = \begin{cases}
\set{u \in N(v): T(u) = T(v^-)} & \text{ if } T(v^-) \in \Nat \\
\set{u \in N(v): T(u) \leq \succb{T(v^-)}} & \text{ if } T(v^-) \in \Natstro
\end{cases}
\end{align}
Second, the definition of the bid \(b^T(v,B)\) is based on \(\tbid{v}\), defined as follows: 
\begin{align}
	\label{eq: bids}
	\tbid{v} = 
	\begin{cases}
		T(v) \ominus T(\vminus) &\text{~if~} T(\vminus) \in \Nat\\
		T(v) \ominus (\absolut{T(\vminus)} + 1) &\text{otherwise}
	\end{cases}
\end{align}

We define the bid chosen by $f_T$ at a configuration $c = \zug{v, B}$. 
Intuitively, \PO ``attempts'' to bid $\tbid{v}$. This is not possible when $\tbid{v}$ requires the advantage but \PO does not have it in $c$, i.e., $\tbid{v} \in \Nat^*\setminus \Nat$ and $B \in \Nat$. In such a case, \PO bids $\absolut{\tbid{v}} + 1 \in \Nat$. Formally, we define $b^T(v,B) = \tbid{v}$ when both $\tbid{v}$ and $B$ belong to either $\Nat$ or $\Nat^*\setminus \Nat$, and $\succb{\tbid{v}}$ otherwise. 

\subsubsection{Strategies that agree with $f_T$ are not losing}
Suppose that the game starts from a configuration $\zug{v, B}$ with $B \geq T(v)$ and \PO follows $f'$ that agrees with \(f_T\). In this section, we will show that $f'$ maintains the invariant that whenever the play reaches a configuration $\zug{u, B'}$, we have $B' \geq T(u)$. This implies that $f'$ is ``not losing'' since if the play reaches a sink $s \in S$, \PO's budget exceeds the frugal target budget. Note that ``not losing'' is not sufficient for winning; indeed, we require a winning strategy to draw the game to a sink. We will show in the next section that there is a winning strategy that agrees with $f_T$. 

We start with the following technical observation. 

\begin{observation}\label{obs: TmatcheswithBid}
	For every vertex $v$, we have \(T(v)\) is in \(\Nat\) iff \(\tbid{v}\) is in \(\Nat\).
\end{observation}

Next, assuming that \PO bids \(\tbid{v}\) from configuration $\zug{v, T(v)}$ (i.e., \PO's budget is $T(v)$), we establish a lower bound on his budget in the next round. The following observation takes care of the case that \PO wins the bidding at $v$. 

\begin{observation}\label{obs: AfterbidBudget}
	If \(T(\vminus) \in \Nat\), then \(T(v) \ominus \tbid{v}  = T(\vminus)\), and if \(T(\vminus) \in \Natstro\), then \(T(v) \ominus \tbid{v} = \absolut{T(\vminus)} + 1\)
\end{observation}

The following lemma takes care of the case that \PO loses the bidding at $v$. 

\begin{lemma}\label{lem:resultonT}
	Let $T$ be a function that satisfies the average property and a vertex $v \in V$. Then \(T(v) \oplus (\succb{\tbid{v}}) = \absolut{T(\vplus)}^*\)
\end{lemma}

\begin{proof}
	We establish this result by analysing four cases in the following, where each case corresponds to a parity of \(\sumT\) and an advantage status of \(T(\vminus)\):
	
	\begin{itemize}
		\item \(\sumT\) is even, and \(T(\vminus) \in \Nat\). 
		
		In this case, \(\tbid{v} = T(v) \ominus T(\vminus) = \frac{\sumT}{2} \ominus T(\vminus) = \frac{\diffT}{2}\). 
		Therefore, we have 
		\begin{align*}
			T(v) \oplus (\succb{\tbid{v}}) = \frac{\sumT}{2} +  \succInt{\frac{\diffT}{2}} = \absolut{T(\vplus)}^*
		\end{align*}
	
	\item \(\sumT\) is odd and \(T(\vminus) \in \Natstro\). 
	
	In this case, \(\tbid{v} = T(v) \ominus (\absolut{T(\vminus)} + 1) = \left(\floor{\frac{\sumT}{2}} + 1\right) - (\absolut{T(\vminus)} + 1) = \floor{\frac{\diffT}{2}}\).
	Therefore, we have
	\begin{align*}
		T(v) \oplus (\succb{\tbid{v}}) = \adjustbrac{\floor{\frac{\sumT}{2}} +1} + \adjustbrac{\floor{\frac{\diffT}{2}}}^* = \absolut{T(\vplus)}^* 
	\end{align*} 

	\item \(\sumT\) is even and \(T(\vminus) \in \Natstro\). 
	
	In this case, \(\tbid{v} = T(v) \ominus (\absolut{T(\vminus)} + 1) = \succInt{\frac{\sumT}{2}} \ominus (\absolut{T(\vminus)} + 1) \linebreak = \succInt{\frac{\sumT}{2} - \absolut{T(\vminus)} - 1} = \succInt{\frac{\diffT}{2} - 1}\). 
	Therefore, we have 
	\begin{align*}
		T(v) \oplus (\succb{\tbid{v}}) &= \succInt{\frac{\sumT}{2}} \oplus \frac{\diffT}{2} = \absolut{T(\vplus)}^*
	\end{align*}

	\item Finally, \(\sumT\) is odd and \(T(\vminus) \in \Nat\). 
	
	In this case, \(\tbid{v} = T(v) \ominus T(\vminus) = \succInt{\floor{\frac{\sumT}{2}}} \ominus T(\vminus) = \succInt{\floor{\frac{\diffT}{2}}}\). 
	Therefore, we have
	\begin{align*}
		T(v) \oplus (\succb{\tbid{v}}) &= \succInt{\floor{\frac{\sumT}{2}}} \oplus \adjustbrac{\floor{\frac{\diffT}{2}} + 1} \\
		&= \succInt{\frac{\sumT}{2} - \half{1} + \half{\diffT} - \half{1} + 1} = \absolut{T(\vplus)}^*
	\end{align*}
	\end{itemize}
\end{proof}

Furthermore, we also make the following observation, which intuitively enlists at least how much \PO's new budget would be after a bidding at \(v\) where he has a budget of \(\succb{T(v)}\) and he bids \(\succb{\tbid{v}}\).

\begin{observation}\label{obs: subsequent}
	Let \(T\) be a function that satisfies the average property and a vertex \(v \in V\), then 
	\begin{itemize}
		\item $(\succb{T(v)}) \ominus (\succb{\tbid{v}}) = T(v) \ominus \tbid{v}$\label{itm: claimd}, and 
		\item $(\succb{T(v)}) \oplus (\tbid{v} \oplus 1) = |T(v^+)|^* + 1$. \label{itm: claime}
	\end{itemize}
\end{observation}

We proceed to prove that strategies that agree with $f_T$ maintain an invariant on \PO's budget. 

\begin{restatable}{lemma}{invariant}
	\label{lem: invariant}
	Suppose that \PO plays according to a strategy $f'$ that agrees with $f_T$ starting from configuration $\zug{v, B}$ satisfying $B \geq T(v)$. 
	Then, against any \PT strategy, when the game reaches \(u \in V\), \PO's budget is at least \(T(u)\).
\end{restatable}

\begin{proof}
The invariant holds initially by the assumption. 
Consider a history that ends in a configuration $\zug{v, B}$. Assume that \(B \geq T(v)\). We claim that the invariant is maintained no matter the outcome of the bidding, namely $B \ominus b^T(v,B) \geq T(v^-)$ and $B \oplus  (\succb{b^T(v,B)}) \geq T(v^+)$.

	We distinguish between two cases. 
	First, when either both \(B\) and \(\tbid{v}\) are in \(\Nat\) or both \(B\) and \(\tbid{v}\) are in  \(\Natstro\). 
	In either case, \PO bids \(b^T(v, B) = \tbid{v}\). 
	It follows from Observation~\ref{obs: AfterbidBudget} and Lemma~\ref{lem:resultonT} that
	\[B \ominus b^T(v, B) = B \ominus \tbid{v} \geq T(v) \ominus \tbid{v} \geq T(v^-), \text{ and}\] 
	\[B \oplus b^T(v, B) \oplus 0^* = B \oplus \tbid{v} \oplus 0^* \geq T(v) \oplus \tbid{v} \oplus 0^* = |T(v^+)|^*.\]

In the second case, \PO bids \(b^T(v, B) = \succb{\tbid{v}}\). 
	Note that  \(B \geq  \succb{T(v)}\) because \(T(v)\) and \(\tbid{v}\) have the same advantage status (Observation~\ref{obs: TmatcheswithBid}). 
It follows from Observation.~\ref{obs: subsequent} that
	\[B \ominus b^T(v, B) = B \ominus (\succb{\tbid{v}}) \geq (\succb{T(v)}) \ominus (\succb{\tbid{v}}) = T(v) \ominus \tbid{v} \geq T(v^-), \text{ and}\] 
	\[B \oplus (\succb{b^T(v, B)}) \geq (\succb{T(v)}) \oplus (\tbid{v} \oplus 1) = |T(v^+)|^* + 1 > |T(v^+)|.\] 
	This concludes the proof. 
\end{proof}

The following proposition follows from Lemma~\ref{lem: invariant}.

\begin{restatable}{proposition}{corinvariant}
	\label{prop: corinvariant}
	Suppose that \PO plays according to a strategy $f'$ that agrees with \(f_T\) starting from configuration \(\zug{v, B}\) satisfying \(B \geq T(v)\), then:  	
	\begin{itemize}
	\item $f'$ is a legal strategy: the bid $b$ prescribed by $f_T$ does not exceed the available budget.
	\item $f'$ does not underestimate the frugal-budget: if \(s \in S\) is reached, \PO's budget is at least~\(\fr(s)\). 
	\end{itemize}
\end{restatable}

\subsubsection{Existence of thresholds in frugal-reachability discrete-bidding games}
We close this section by showing existence of threshold budgets in frugal-reachability discrete-bidding games. 
Recall that Theorem~\ref{thm:non-unique} shows that functions that satisfy the discrete average property are not unique. 
Let \(T\) be such a function.
The following lemma shows that \(T \leq \thresh_{\G}\). 
That is, if in a vertex \(v\), \PO has a budget less than \(T(v)\), then \PT has a winning strategy.
This proves that the threshold budgets for \PO cannot be less than \(T(v)\), when \(T\) is a function that satisfies the average property.

\begin{restatable}{lemma}{reachlower}
	\label{lem:reach-lower}
	Consider a frugal-reachability discrete-bidding game \(\G = \zug{V, E, k, S, \fr}\). If \(T: V \rightarrow [k] \cup \{k+1\}\) is a function that satisfies the average property, then \(T(v) \leq \thresh_{\G}(v)\) for every~\(v \in V\). 
\end{restatable}

\begin{proof}

Given $T$ that satisfies the average property, we construct $T'$ as in Definition~\ref{def: ComplimentofT}. 
Let $\zug{v, B_1}$ be a configuration, where $v \in V$, \PO's budget is $B_1$, and implicitly, \PT's budget is $B_2 = k^* \ominus B_1$. Note that $B_1 < T(v)$ iff $B_2 \geq T'(v)$. Moreover, for every $s \in S$, we have $T'(s) = (k+1) \ominus \fr(s)$. 
We consider the ``flipped'' game; namely, we associate \PT with \PO (of Proposition~\ref{prop: corinvariant}), and construct a partial strategy $f_{T'}$ for \PT. 
 We construct a \PT strategy $f'$ that agrees with $f_{T'}$: for each $v \in V$, we arbitrarily choose a neighbor $u$ from the allowed vertices. By Lemma~\ref{lem: invariant}, no matter how \PO responds, whenever the game reaches $\zug{u, B_1}$, we have $B_2 \geq T'(u)$. The invariant implies that $f'$ is a winning strategy. Indeed, if the game does not reach a sink, \PT wins, and if it does, \PO's frugal objective is not satisfied. 
\end{proof}

The following lemma shows the existence of a function that satisfies the average property and that coincides with threshold budgets.

\begin{restatable}{lemma}{reachupper}
\label{lem:reach-upper}
Consider a frugal-reachability discrete-bidding game $\G = \zug{V, E, k, S, \fr}$. There is a function $T$ 
that satisfies the average property with $T(v) \geq \thresh_\G(v)$, for every $v \in V$.
\end{restatable}

\begin{proof}
	
	The proof is similar to the one in \cite{DP10}. We illustrate the main ideas. For $n \in \Nat$, we consider the {\em truncated game} $\G[n]$, which is the same as $\G$ except that \PO wins iff he wins in at most $n$ steps. We find a sufficient budget for \PO to win in the vertices in $\G[n]$ in a backwards-inductive manner. For the base case, for every vertex $u \in V \setminus S$, since \PO cannot win from $u$ in $0$ steps, we have $T_0(u) = k+1$. For $s \in S$, we have $T_0(s) = \fr(s)$. Clearly, $T_0 \nequiv \thresh_{\G[0]}$. For the inductive step, suppose that $T_{n-1}$ is computed. For each vertex $v$, we define $T_n(v) = \floor{\frac{|T_{n-1}(v^+)| + |T_{n-1}(v^-)|}{2}} +\eps$ as in Def.~\ref{def:average}. Following a similar argument to Theorem~\ref{thm:cont-reach}, it can be shown that if \PO's budget is $T_n(v)$, he can bid a $b$ so that if he wins the bidding, his budget is at least $T_{n-1}(v^-)$ and if he loses the bidding, his budget is at least $T_{n-1}(v^+)$. By induction, we get $\thresh_{\G[n]}(v) = T_{n}(v)$, for every $v \in V$, which also implies that \(T_n\) is a monotonically non-increasing function. Thus, for every vertex $v$, we let $T(v) = \lim_{n \to \infty} T_n(v)$, which is well-defined because of the monotonicity of \(T_n\) (which is coming from monotonicity of \(\thresh_{\G[n]}\) and the fact that \(\thresh_{\G[n]} = T_n\)), and the fact that it only takes finitely many values, namely ranging over \([k] \cup \{k+1\}\).
	It is not hard to show that $T$ satisfies the average property and that $T(v) \geq \thresh_\G(v)$, for every $v \in V$. 
\end{proof}

Let $T$ be a function that results from the fixed-point computation from the proof of Lemma~\ref{lem:reach-upper}. Since it satisfies the average property, we apply Lemma~\ref{lem:reach-lower} to show that \PT wins from $v$ when \PO's budget is $\predb{T(v)}$. 
Since the values observed in a vertex during an execution of the fixed-point algorithm are monotonically decreasing and since the number of values that a vertex can obtain is $2k +1$, the running time is $O(|V| \cdot k)$. 
We thus conclude the following.

\begin{restatable}{theorem}{thm:frugal-reach}
	\label{thm:frugal-reach}
	Consider a frugal-reachability discrete-bidding game \(\G = \zug{V, E, k, S, \fr}\). Threshold budgets exist and satisfy the average property. Namely, there exists a function \(\thresh: V \rightarrow [k] \cup \{k+1\}\) such that for every vertex $v \in V$ 
	\begin{itemize}
		\item if \PO's budget is $B \geq \thresh(v)$, then \PO wins the game, and
		\item if \PO's budget is $B < \thresh(v)$, then \PT wins the game
	\end{itemize}
	Moreover, there is an algorithm to compute $\thresh$ that runs in time $O(|V| \cdot k)$, which is exponential in the size of the input when $k$ is given in binary.
\end{restatable}

A frugal-safety objective is dual to a frugal-reachability objective. 
 Thus, if \(\thresh: V \rightarrow [k] \cup \{k+1\}\) is the function providing the threshold budgets for the reachability player, then the complement of \(\thresh\) (Definition~\ref{def: ComplimentofT}), denoted by \(\invT{\thresh}\), provides the threshold budget for the safety player. 
 Therefore, we conclude the following 
  
\begin{corollary}\label{corr: frugalSafety}
	Consider a frugal-safety discrete-bidding game \(\G = \zug{V, E, k, S, \fr}\). Threshold budgets exist and satisfy the average property. 
	Namely, there exists a function \(\thresh: V \rightarrow [k] \cup \{k+1\}\) such that for every vertex $v \in V$ 
	\begin{itemize}
		\item if \PO's budget is $B \geq \thresh(v)$, then \PO wins the game, and
		\item if \PO's budget is $B < \thresh(v)$, then \PT wins the game
	\end{itemize}
	Moreover, there is an algorithm to compute $\thresh$ that runs in time $O(|V| \cdot k)$, which is exponential in the size of the input when $k$ is given in binary.
\end{corollary}

\begin{remark}
We point to a conceptual similarity with {\em concurrent stochastic games}~\cite{Everett57}, in which in each turn, both players concurrently choose actions, and the joint action gives rise to a probability distribution over next states. There, the value of the game is given by the least fixed point of what is referred to as a \emph{value mapping} function~\cite{Everett57,FV12}. 
One can adapt their operator to our setting, thereby obtaining a operator that maps functions from states to budgets. Intuitively, each such function would be a candidate for the thresholds, and functions that satisfy the average property are the fixed points. As we show above, the maximal fixed point is the function that coincides with the thresholds. 
\stam{
On the other hand, structurally there is a similarity in between our algorithm for computing the threshold budgets and how values of a concurrent stochastic reachability games are computed~\cite{Everett57, FV12}. 
There, the values are given by the least fixed point of what is referred to as a \emph{value mapping} function~\cite{Everett57}. 
In our case this value mapping function maps one candidate for a function computing threshold budgets to another candidate of the same.
}
\end{remark}

\section{A Fixed-Point Algorithm for Finding Threshold Budgets}
In this section, we develop a fixed-point algorithm for finding threshold budgets in frugal-parity discrete-bidding games. While its worst-case running time is exponential in the input, the algorithm shows, for the first time, that threshold budgets in parity discrete-bidding games satisfy the average property. This property is key in the development of the NP and coNP algorithm (Sec.~\ref{sec:NP and coNP}).

\subsection{Warm up: a fixed-point algorithm for B\"uchi bidding games}
\label{sec:Buchi}
In this section, we illustrate the ideas of the fixed-point algorithm on the special case of B\"uchi games. The transition from B\"uchi to parity games involves an induction on the parity indices. A B\"uchi game is $\zug{V, E, k, F}$, where $F \subseteq V$ is a set of {\em accepting states}. Formally, B\"uchi games are a special case of parity games in which the vertices are labeled by $\set{0,1}$. We stress that, for ease of presentation, we focus on games without a frugal objective.

Throughout this section we will take the perspective of \PO, the co-B\"uchi player, whose objective is to visit $F$ only finitely often. We present an algorithm that takes as input a B\"uchi game $\G$ and outputs \PO's thresholds, which we denote by $\coTh$. That is, for an initial configuration $\zug{v, B}$, we have:
\begin{itemize}
\item if $B \geq \coTh(v)$, \PO can guarantee that $F$ is visited only {\em finitely} often, and
\item if $B < \coTh(v)$, \PT can guarantee that $F$ is visited {\em infinitely} often.
\end{itemize}

The fixed-point algorithm repeatedly finds thresholds in a sequence of increasingly easier objectives (from \PO's perspective) whose limit is the co-B\"uchi objective. 
Roughly, recall that the co-B\"uchi objective requires that $F$ is eventually not visited without specifying a bound on the number of visits to $F$. The objective $\safe_i$, for $i \geq 0$, introduces a restriction: $F$ can be visited at most~$i$ times. In particular, $\safe_0$ is a safety game. Formally, 

\begin{definition}[Bounded-eventual safety objectives]
For $i \geq 0$, the objective $\safe_i$ contains infinite paths that 
\begin{itemize}
\item start in $V \setminus F$ and enter $F$ at most $i$ times before exiting \(F\) eventually, or
\item start in $F$, exit $F$ for the first time at some point, and then enter $F$ at most $i-1$ more times before eventually exiting \(F\)
\end{itemize}
\end{definition}
The formal definition of when a path enters $F$ can be found in Sec.~\ref{sec:prelim-obj}. Note that we define $\safe_0$ so that every path that starts in $F$ violates $\safe_0$.

For $i \geq 0$, we denote by $\thresh_i: V \rightarrow [k] \cup \{k+1\}$ the threshold for the objective $\safe_i$. We make two observations. 

\begin{observation}
\label{obs:Buchi}
For $i \geq 0$ and $v \in \notF$, we have:
\begin{itemize}
\item $\thresh_i(v) \geq \coTh(v)$, and  
\item $\thresh_i(v) \geq \thresh_{i+1}(v)$.
\end{itemize}
\end{observation}
\begin{proof}
First, ensuring the \coBuchi objective (i.e., entering $F$ only finitely often) is easier than ensuring $\safe_i$ (i.e., entering $F$ at most $i$ times), meaning that more budget is necessary for $\safe_i$ than for \coBuchi, thus $\thresh_i(v) \geq \coTh(v)$. Second, similarly, since the restriction imposed by $\safe_i$ is harder than the restriction imposed by $\safe_{i+1}$, more budget is required for the latter, thus $\thresh_i(v) \geq \thresh_{i+1}(v)$.
\end{proof}

It follows that the sequence $\thresh_0,\thresh_1,\ldots$ of thresholds reaches a fixed point. We will show that the fixed point coincides with $\coTh$.

\subsubsection{A recursive algorithm to compute thresholds for $\safe_i$}
We describe a recursive algorithm to compute $\thresh_i$, for $i \geq 0$. The idea is to characterize $\thresh_i$ as thresholds in two bidding games: a frugal-reachability game $\R_i$ and a frugal-safety game $\S_i$. Throughout this section, we follow the convention of using $v$ to denote a vertex in $V \setminus F$ and $u$ to denote a vertex in $F$. 

\paragraph*{Base case.}
Recall that $\safe_0$ is a safety objective: \PO wins by ensuring that $F$ is not visited at all. In particular, paths that start from $F$ are losing for \PO, thus as in Remark~\ref{rem:losing-vertices}, we have the following.

\begin{lemma}
For $u \in F$, we have $\thresh_0(u) = k+1$. 
\end{lemma}

\paragraph*{Recursive step.}

Note that, in the base case, we have only computed \(\thresh_0(u)\) for \(u \in F\), and not \(\thresh_0(v)\) for \(v \in V \setminus F\). 
Therefore, in the general recursive step, we assume \(\thresh_i(u)\) have already been computed for every \(u \in F\) (recursive step hypothesis) and here we first compute \(\thresh_i(v)\) for \(v \in V \setminus F\), which is followed by the computation of \(\thresh_{i+1}(u)\) for \(u \in F\). 
 
We first characterize the thresholds in vertices in $\notF$ as thresholds in a frugal-safety game. 
Let $i \geq 0$ and suppose that $\thresh_i$ has been computed for vertices in $F$. Recall that for $u \in F$, a budget of $\thresh_i(u)$ is the threshold to ensure the objective of exiting $F$ and visiting $F$ at most $i-1$ more times. Suppose that the game starts in $v \in \notF$. \PO wins by either not visiting $F$ at all, or if $u \in F$ is reached, \PO's budget should exceed $\thresh_{i-1}(u)$ since he can continue with a strategy that ensures $\safe_{i-1}$. Formally, we characterize $\thresh_i$ in $\notF$ as thresholds in a frugal-safety game played on the same arena as $\G$. 

\begin{lemma}\label{lem: frugalsafetytoSafei}
Consider the frugal-safety game $\S_i = \zug{V, E, k, F, \fr_i}$, where every vertex in \(F\) are sinks, and for each $u \in F$, we have $\fr_i(u) = \thresh_{i}(u)$. 
Then, for every $v \in \notF$, we have $\thresh_i(v) = \thresh_{\S_i}(v)$. 
\end{lemma}

\begin{remark}\label{rem: averagepropertyinherit}
	Note that, in \(\S_i\), each vertex \(u \in F\) are sinks. Thus, \PO's threshold budget at those vertices in that game would be the same as the frugal budget. That is, for \(u \in F\), \(\thresh_{\S_i}(u) = \fr_i(u) = \thresh_i(u)\). Therefore, we indeed have for every vertex \(w \in V\) of \(\S_i\), \(\thresh_i(w) = \thresh_{S_i}(w)\). 
	Since \(\thresh_{S_i}\) satisfies the average property (Corollary~\ref{corr: frugalSafety}), so does \(\thresh_i\). 
\end{remark}

Next, we characterize the thresholds in vertices in $F$ as thresholds in a frugal-reachability game. 
We assume that $\thresh_{i}$ has been computed for vertices in $\notF$. 
Recall that for a vertex $v \in \notF$, a budget of $\thresh_{i}(v)$ is the threshold to ensure the objective of visiting $F$ at most $i$ times. 
Suppose that the game starts in $u \in F$. 
In order to ensure $\safe_{i+1}$, \PO must first force the game out of $F$ and then ensure that $F$ is visited at most $i$ more times. This is achieved by ensuring that some $v \in \notF$ is reached with a budget of $\thresh_{i}(v)$. Formally, we characterize $\thresh_{i+1}$ in $F$ as thresholds in a frugal-reachability game played on the same arena as $\G$.

\begin{lemma}\label{lem: frugalreachtoSafei}
	Consider the frugal-reachability game $\R_{i+1} = \zug{V, E, k, \notF, \fr_{i+1}}$, where every vertex in \(V \setminus F\) are sinks and for each $v \in \notF$, we define $\fr_{i+1}(v) = \thresh_{i}(v)$. 
	Then, for every $u \in F$, we have $\thresh_{i+1}(u) = \thresh_{\R_{i+1}}(u)$. 
\end{lemma}

\paragraph*{Pseudocode} 
We conclude this section with a pseudocode of the fixed-point algorithm. See Fig.~\ref{fig:Buchi} for a depiction of its operation. 
The algorithm calls two sub-routines $\Calli{Frugal-Reachability}$ and $\Calli{Frugal-Safety}$, which return the thresholds for all vertices that are not targets, respectively, in a frugal-reachability and frugal-safety game, e.g., by running the fixed-point algorithm described in Lemma~\ref{lem:reach-upper}.



\let\oldnl\nl
\newcommand{\nonl}{\renewcommand{\nl}{\let\nl\oldnl}}
\SetKwProg{Prog}{Algorithm}{}{}
\SetKwRepeat{Dowhile}{do}{while}%
\SetKwComment{Comment}{}{}
\SetProcNameSty{textsc}
 \begin{algorithm}[ht]
\caption{A fixed-point algorithm to find threshold budgets in co-B\"uchi games.}
\label{alg:Frugal-Buchi}
\nonl\Prog{co-B\"uchi-Thresholds($\G$)}{}
 $i := 0$\;
 Define the frugal-safety game $\S_0 = \zug{V, E, k, F, \fr_0}$, with $\fr_0 \equiv k+1$.\;
 $\thresh_{\S_0} = \Call{Frugal-Safety}{\S_0}$\;
 Define $\thresh_0(v) = \thresh_{\S_0}(v)$, for $v \in V \setminus F$, and $\thresh_0(u) = k+1$, for $u \in F$.\;
\Dowhile{$\thresh_{i-1} \neq \thresh_i$}{
 $i := i+1$\;
 Define $\R_i = \zug{V, E, k, V \setminus F, \thresh_{\S_{i-1}}}$.\;
 $\thresh_{\R_i} := \Call{Frugal-Reachability}{\R_i}$ \Comment{Thresholds for vertices in $F$.}
 Define $\S_i = \zug{V, E, k, F, \thresh_{\R_i}}$\;
 $\thresh_{\S_i} := \Call{Frugal-Safety}{\S_i}$ \Comment{Thresholds for vertices in $V \setminus F$.}
 Define $\thresh_i(v) = \thresh_{\S_i}(v)$, for $v \in V \setminus F$, and $\thresh_i(u) = \thresh_{\R_i}(u)$, for $u \in F$.\;
 }
\end{algorithm}

\begin{figure}
  \begin{subfigure}[t]{0.95\textwidth}
    \begin{center}
       	\includegraphics[width=0.6\textwidth]{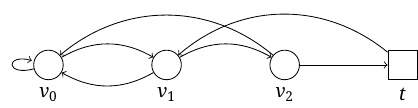}
    \end{center}

			
	\end{subfigure}

	\bigskip
	
	\begin{subfigure}[t]{0.95\textwidth}
          \begin{center}
             	\includegraphics[width=0.9\textwidth]{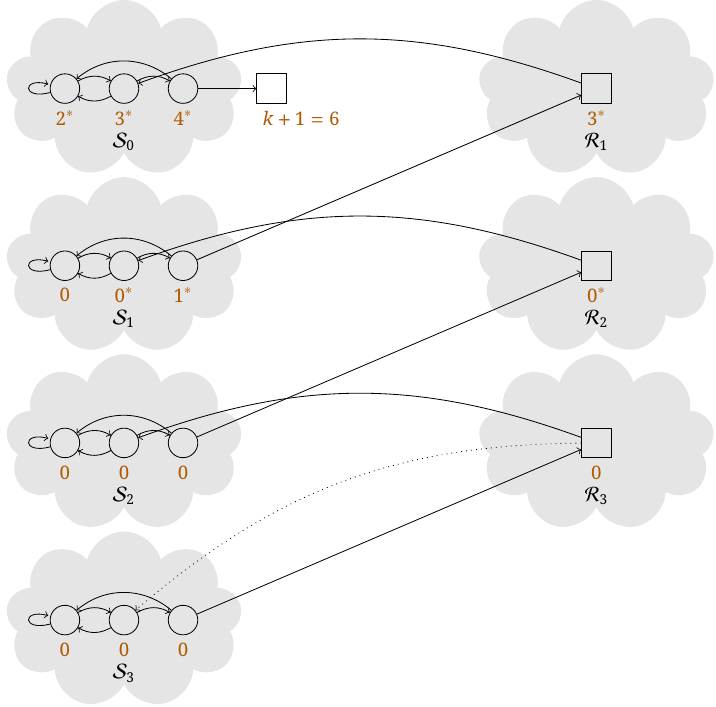}
              \end{center}
	\end{subfigure}
	\caption{A depiction of the fixed-point algorithm for the \coBuchi game on top of the figure. The goal of \PO (the \coBuchi player), is to visit $t$ only finitely often. Set the total budget to $k=5$. The lower part of the figure depicts the progress of thresholds, depicted in orange.
	For example, the game $\S_1$ is a frugal-safety game in which the safety player wins either if the game never reaches $t$ or if it reaches $t$, his budget is at least $3^*$. The frugal-reachability games have a trivial state space. The threshold in $t$ in $\R_i$ coincides with $v_1$ in $\S_{i-1}$ since the reachability player guarantees reaching the target $v_1$ with a sufficient budget by bidding $0$.
The algorithm terminates once a fixed-point is reached. In this example, we see that the thresholds in \(\G\) are \(0\) in all vertices.}
\label{fig:Buchi}
\end{figure}

\subsubsection{The fixed point coincides with $\coTh$}
As mentioned above, the sequence $\thresh_0,\thresh_1,\ldots$ of thresholds reaches a fixed point. We show that the fixed-point threshold coincides with the co-B\"uchi thresholds. 

\begin{theorem}
\label{thm:correctness-Buchi}
Consider a B\"uchi bidding game $\G$. For $i \geq 0$, let $\thresh_i$ be the threshold for satisfying the objective $\safe_i$, and $n \in \Nat$ be such that $\thresh_n = \thresh_{n+1}$. Then, $\thresh_n$ coincides with the thresholds for the co-B\"uchi player, namely $\thresh_n = \coTh$. Moreover, $\coTh$ satisfies the average property and computing it can be done in time $O\big((|V| \cdot k)^2\big)$.
\end{theorem}
\begin{proof}
We show that $\thresh_n = \coTh$. First, $\thresh_n \geq \coTh$ follows immediately from Observation~\ref{obs:Buchi}; recall that $\thresh_n$ is the threshold for the objective $\safe_n$,  which is harder for \PO to ensure than the \coBuchi objective. 

Second, we show that $\thresh_n \leq \coTh$. 
Consider a vertex $v \in \notF$. We show that \PT, the B\"uchi player, wins from a configuration $\zug{v, B}$ with $B < \thresh_n(v) = \thresh_{n+1}(v)$. The case of an initial vertex in $F$ is also captured in the following proof. \PT plays as follows from $\zug{v, B}$. Recall that $\thresh_{n+1}(v) = \thresh_{\S_{n+1}}(v)$. Thus, a budget of $B < \thresh_{\S_{n+1}}(v)$ means that the reachability player wins the frugal-safety game $\S_i$. \PT follows the reachability player's strategy to ensure that $F$ is eventually reached with a budget below the frugal-target budget. Formally, a configuration $\zug{u, B'}$ is reached with $u \in F$ and $B' < \thresh_{\R_{n}}(u)$. Next, a budget of $B'$ suffices for the safety player to win the frugal-reachability game $\R_i$ and ensures that either: (1)~$F$ is never exited, or (2)~$\notF$ is visited with a budget that violates the frugal-target budget. \PT follows such a winning strategy to ensure either (1)~$F$ is never exited thus is clearly visited infinitely often or (2)~a configuration $\zug{v', B''}$ is reached with $v' \in \notF$ and $B'' < \thresh_{\S_{n}}(v')$. Since $\thresh_{\S_{n}}(v') = \thresh_{\S_{n+1}}(v')$, \PT restarts her strategy from $v'$. Thus, in both cases \PT guarantees infinitely many visits to $F$, and we are done. 

Finally, the thresholds satisfy the average property since so does each \(\thresh_i\) (see Remark~\ref{rem: averagepropertyinherit}). 
Regarding running time, note that the thresholds observed in a vertex are monotonically decreasing (Observation~\ref{obs:Buchi}), thus the number of iterations until a fixed-point is reached is $O(|V| \cdot k)$. Each iteration includes two solutions of frugal-reachability games, each of running time $O(|V| \cdot k)$ (Theorem~\ref{thm:frugal-reach}). 
\end{proof}

\subsection{A fixed-point algorithm for frugal-parity bidding games}
In this section, we extend the fixed-point algorithm developed in Sec.~\ref{sec:Buchi} to parity bidding games. The algorithm involves a recursion over the parity indices, which we carry out by strengthening the induction hypothesis and developing an algorithm for frugal-parity objectives instead of the special case of parity objectives. 

For the remainder of this section, fix a frugal-parity game $\G = \zug{V, E, k, p, S, \fr}$. Denote the maximal parity index by $d \in \Nat$. 
Recall that $S$ is a set of sinks where the parity indices are not defined and $\fr(s)$ denotes \PO's frugal target budget at $s \in S$. Thus, \PO wins a play $\pi$ if 
\begin{itemize}
\item $\pi$ is infinite and satisfies the parity condition, or 
\item $\pi$ is finite and ends in a configuration $\zug{s, B}$ with $s \in S$ and $B \geq \fr(s)$.
\end{itemize}

We characterize the thresholds in $\G$ by reasoning about games with a lower parity index. This characterization gives rise to a recursive algorithm to compute the thresholds.

\begin{lemma}[Base case]
Let $\G = \zug{V, E, p, S, \fr}$ with only one parity index, i.e., $p(v) = d$, for all $v \in V$ and \(S\) is the set of sinks for which frugal-budgets are given by \(\fr\). 
\begin{itemize}
\item Assume that $d$ is odd. Let $\S = \zug{V, E, S, \fr}$ be a frugal-safety game. Then, $\thresh_\G \equiv \thresh_\S$. 
\item Assume that $d$ is even. Let $\R = \zug{V, E, S, \fr}$ be a frugal-reachability game. Then, $\thresh_\G \equiv \thresh_\R$. 
\end{itemize}
\end{lemma}
\begin{proof}
Clearly, in both cases, a finite play that ends in a sink is winning in $\G$ iff it is winning in $\S$, and similarly for $\R$. 
When $d$ is odd, any infinite play in $\G$ is winning for \PO, thus $\G$ is a frugal-safety game. On the other hand, when $d$ is even, any infinite play in $\G$ is losing for \PO, and the only way to win is by satisfying the frugal objective in a sink, thus $\G$ is a frugal-reachability game. 
\end{proof}

\begin{corollary}
When $\G$ contains only one parity index, computing $\thresh_\G$ can be done by calling a sub-routine that finds the thresholds in a frugal-reachability (or a frugal-safety) bidding game. Moreover, by Theorem~\ref{thm:frugal-reach}, $\thresh_\G$ satisfies the average property in this case. 
\end{corollary}

\paragraph*{Recursive step}
Suppose that more than one parity index is used. Let $d \in \Nat$ denote the maximal parity index in $\G$. 
We assume access to a sub-routine that computes thresholds in frugal-parity games with a maximal parity index of $d-1$, and we describe how to use it in order to  compute thresholds in $\G$. We assume that $d$ is even, and we describe the algorithm from \PO's perspective. 
The definition for an odd $d$ is dual from \PT's perspective.

Let $F_d = \set{v: p(v) = d}$. Since $d$ is even, a play that visits $F_d$ infinitely often is losing for \PO. Thus, a necessary (but not sufficient) requirement to win is to ensure that $F_d$ is visited only finitely often. For example, a B\"uchi game can be modeled as follows: \PO is the co-B\"uchi player, the parity indices are $1$ or $2$, and the set $F_2$ denotes the accepting vertices, which \PO needs to visit only finitely often. 

We define a bounded variant of the frugal-parity objective, similar to the definition of $\safe_i$ in Sec.~\ref{sec:Buchi}:
\begin{definition}
For $i \geq 0$, a play $\pi$ is in $\frparity_i$ if:
\begin{itemize}
\item $\pi$ is finite and satisfies the frugal objective: ends in $\zug{s, B}$ with $s \in S$ and $B \geq \fr(s)$, or
\item $\pi$ is infinite, satisfies the parity objective, and 
\begin{itemize}
\item starts from $\notFd$ and enters $F_d$ at most $i$ times before eventually exiting, or
\item starts from $F_d$, exits $F_d$ for the first time at some point, and then enters $F_d$ at most $i-1$ more times before exiting eventually. 
\end{itemize}
\end{itemize}
\end{definition}
In particular, every path that starts from $F_d$ violates $\frparity_0$. 

For $i \geq 0$, we denote by $\thresh_i$ the thresholds for objective $\frparity_i$. As in Observation~\ref{obs:Buchi}, since the restriction monotonically decreases as $i$ grows, the thresholds are monotonically non-increasing and they all lower-bound the thresholds in $\G$.

\begin{observation}
\label{obs:parity}
For $i \geq 0$, we have $\thresh_{i+1} \leq \thresh_i$ and $\thresh_{\G} \leq \thresh_i$. 
\end{observation}
It follows that the sequence of thresholds reaches a fixed-point, and we will show that the thresholds at the fixed point coincide with $\thresh_\G$. 

We iteratively define and solve two sequences of games: a sequence of frugal-parity games $\G_0,\G_1,\ldots$ each with maximal parity index $d-1$ and a sequence of frugal-reachability games $\R_0,\R_1,\ldots$. For $i \geq 0$, recall that $\thresh_{\G_i}$ and $\thresh_{\R_i}$ respectively denote the thresholds in $\G_i$ and $\R_i$. We will show that $\thresh_i$ can be characterized by $\thresh_{\G_i}$ and $\thresh_{\R_i}$: we will show that for $v \in F_d$ we have $\thresh_i(v) = \thresh_{\G_i}(v)$ and for $u \in \notFd$, we have $\thresh_i(u) = \thresh_{\R_i}(u)$. 

We start with the frugal-parity games. The games share the same arena, which is obtained from $\G$ by setting the vertices in $F_d$ to be sinks. 
The games differ in the frugal target budgets. Formally, for $i \geq 0$, we define $\G_i = \zug{V, E', p', S', \fr_{\G_i}}$, where the sinks are $S' = S \cup F_d$, the edges are restricted accordingly $E' = \set{\zug{v, v'} \in E: v \in \notFd}$, the parity function $p'$ coincides with $p$ but is not defined over $F_d$, i.e., $p'(v) = p(v)$ for all $v \in V \setminus F_d$, and $\fr_{\G_i}$ is the only component that changes as \(i\) changes, and it is defined below based on a solution to $\R_i$.   Note that $p'$ assigns at most $d-1$ parity indices.

We construct the frugal-reachability games. Let $i \geq 0$. Intuitively, the game $\R_i$ starts from $F_d$ and \PO's goal is to either satisfy the frugal objective in $S$ or reach $\notFd$ with a budget that suffices to ensure that $F_d$ is entered at most $i$ more times. Formally, we construct the frugal-reachability game $\R_i = \zug{V, E'', \notFd \cup S, \fr_{\R_i}}$, where $E'' = \set{\zug{u, u'} \in E: u \in F_d}$ and 
\[
\fr_{\R_i}(v) = \begin{cases}\fr(v) & \text{ if } v \in S\\ \thresh_{\G_i}(v) & \text{ if } v \in \notFd. \end{cases}
\]


\begin{lemma}
\label{lem:parity-FrReach}
Let $i \geq 0$. Assume that for every $v \in V \setminus F_d$, we have $\thresh_{\G_i}(v) = \thresh_i(v)$. Then, for every $u \in F_d$, we have $\thresh_i(u) = \thresh_{\R_i}(u)$. 
\end{lemma}
\begin{proof}
Suppose that $\G$ starts from $\zug{u, B}$ with $u \in F_d$. We first show that when $B \geq \thresh_{\R_i}(u)$, \PO can ensure the objective $\frparity_i$. Indeed, by following a winning strategy in $\R_i$, \PO guarantees that either (1) the frugal objective is satisfied in $S$, in which case the play is clearly winning in $\G$, or (2) the game reaches a configuration $\zug{v, B'}$ with $v \in \notFd$ and $B' \geq \thresh_{\G_i}(v)$, from which, by the assumption that $\thresh_{\G_i}(v) = \thresh_i(v)$, he can proceed with a winning strategy for $\frparity_{i}$. 

On the other hand, when $B < \thresh_{\R_i}(u)$, \PT violates $\frparity_i$ as follows. She first follows a winning strategy in $\R_i$, which ensures that no matter how \PO plays, the resulting play either (1) violates the frugal objective in $S$, (2) stays in $F_d$, or (3) it reaches a configuration $\zug{v, B'}$ with $v \in \notFd$ and $B' < \thresh_{\G_i}(v) = \thresh_i(v)$. In Cases~(1) and~(2), the play is clearly winning for \PT for violating the objective \(\frparity_{i}\) in \(\G\), and in Case~(3), the assumption on $\thresh_i(v)$ implies that \PT can continue with a strategy that violates $\frparity_{i}$.
\end{proof}

\begin{remark}
	Similar to Remark~\ref{rem: averagepropertyinherit}, here too, we obtain that \(\thresh_i\) satisfies the average property because it coincides with \(\thresh_{R_i}\) which is a function providing threshold budgets in a frugal-reachability game (Theorem~\ref{thm:frugal-reach}).   
\end{remark}

We define the frugal target budgets $\fr_{\G_i}$ of the frugal-parity game $\G_i$. Recall that we obtain $\G_i$ from $\G$ by setting $F_d$ to be sinks. Thus, the sinks in $\G_i$ consist of ``old'' sinks $S$ and ``new'' sinks $F_d$. The frugal target budgets of $\G$ and $\G_i$ agree on $S$, thus for $s \in S$ and $i \geq 0$, we have $\fr_{\G_i}(s) = \fr(s)$. For $u \in F_d$, we define $\fr_{\G_0}(u) = k+1$ and for $i > 0$, we define $\fr_{\G_i}(u) = \thresh_{\R_{i-1}}(u)$. 

\begin{lemma}
	\label{lem:FrParity-i}
	For \(i \geq 0\) and \(u \in F_d\), assume that a budget of \(\fr_{\G_i}(u)\) is the threshold to satisfy \(\frparity_{i}\). Then, for \(v \in \notFd\), we have \(\thresh_i(v) = \thresh_{\G_i}(v)\). 
\end{lemma}


\begin{proof}
	Recall that each $\G_i$ agrees with $\G$ on the parity indices in $\notFd$, thus an infinite path that satisfies the parity condition in $\G_i$ satisfies it in $\G$, and that $\G_i$ and $\G$ agree on the frugal target budgets in $S$. 
	
	Under the assumption in the statement, we prove that $\thresh_{i}(v) = \thresh_{\G_i}(v)$, for $v \in \notFd$. 
	Suppose that $\G$ starts from $\zug{v, B}$ with $v \in \notFd$. First, when $B \geq \thresh_{\G_i}(v)$, \PO ensures $\frparity_i$ by following a winning strategy in $\G_i$. Let $\pi$ be the play that is obtained when \PT follows some strategy. Note that $\pi$ is winning for \PO in $\G_i$, thus it satisfies one of the following: 
	\begin{enumerate}
		\item $\pi$ is finite and ends in $\zug{s, B}$ with $s \in S$ and $B \geq \fr_{\G_i}(s) = \fr(s)$, 
		\item $\pi$ is infinite (i.e., a sink is never reached) and satisfies the parity condition,  or 
		\item $\pi$ is finite and ends in $\zug{u, B}$ with $u \in F_d$ and $B \geq \fr_{\G_i}(u)$. 
	\end{enumerate} 
	Case~(1) clearly satisfies the frugal objective of $\frparity_i$, in Case~(2) the parity condition is satisfied without visiting $F_d$ once, thus again, $\frparity_i$ is satisfied. 
	Finally, in Case~(3), once the game reaches $\zug{u, B}$, the assumption on \(\fr_{\G_i}(u)\) implies that \PO can follow a strategy that ensures \(\frparity_{i}\). 
	Second, if $B < \thresh_{\G_i}(v)$, \PT violates $\frparity_i$ by following a winning strategy in $\G_i$. The argument is dual to the above. 
\end{proof}

Note that since every path that starts from $F_d$ violates $\frparity_0$, the threshold budget at every $u \in F_d$ is $k+1$. This constitutes the proof of the base case of the following lemma, and the inductive step is obtained by combining Lemma~\ref{lem:parity-FrReach} with Lemma~\ref{lem:FrParity-i}.

\begin{lemma}
\label{lem:frParityUpper}
For $i \geq 0$, for $v \in V$ we have $\thresh_i(v) = \thresh_{\G_i}(v)$ and for $u \in \notFd$, we have $\thresh_i(u) = \thresh_{\R_i}(u)$. 
\end{lemma}

It follows from Observation~\ref{obs:parity} that the sequence $\thresh_0,\thresh_1,\ldots$ reaches a fixed point. We show that at the fixed point, the threshold coincides with $\thresh_\G$. 

\begin{lemma}
Let $n \in \Nat$ such that $\thresh_n = \thresh_{n+1}$. Then, $\thresh_\G = \thresh_n$. 
\end{lemma}
\begin{proof}
Lemma~\ref{lem:frParityUpper} and Observation~\ref{obs:parity} show that $\thresh_\G \leq \thresh_n$. To show equality, we show that \PT wins $\G$ starting from a configuration $\zug{v, B}$ with $v \in \notFd$ and $B < \thresh_n(v)$. \PT proceeds by following a winning strategy in $\G_{n+1}$. Let $\pi$ be a play that results from some \PO strategy. Since $\pi$ is winning for \PT in $\G_{n+1}$,  there are three cases: 
\begin{enumerate}
\item $\pi$ is finite and ends in $\zug{s, B'}$ with $s \in S$ and $B' < \fr(s)$, thus it is winning also in $\G$, 
\item $\pi$ is infinite and violates the parity objective, thus since $\G$ and $\G_{n+1}$ agree on the parity indices, $\pi$ is winning for \PT in $\G$, or 
\item $\pi$ ends in $\zug{u, B'}$ with $B' < \fr_{n+1}(u)$. 
\end{enumerate}
In Case~(3), since $\fr_{n+1}(u) = \thresh_{\R_n}(u)$, \PT continues by following a winning strategy for the safety player in $\R_n$. This guarantees that no matter how \PO plays, the play either stays within $F_d$, thus it necessarily violates the parity objective of $\G$, or it reaches $\zug{v, B''}$ with $v \in \notFd$ and $B'' < \fr_{\R_n}(v)$.  In the latter case, since $\fr_{\R_n}(v) = \thresh_n(v) = \thresh_{n+1}(v)$, \PT can restart her strategy. Note that \PT's strategy guarantees that either $F_d$ is eventually never reached, then she wins, or it is reached infinitely often, in which case she also wins since the play visits parity index $d$ infinitely often. 
\end{proof}

\paragraph*{Pseudocode}

The algorithm is described in Alg.~\ref{alg:Frugal-parity} for an even $d$ and from \PO's perspective.



\begin{algorithm}[ht]
  \setcounter{AlgoLine}{0}
\caption{A fixed-point algorithm to find threshold budgets in frugal parity games.}
\label{alg:Frugal-parity}
\nonl\Prog{Frugal-Parity-Threshold($\G = \zug{V, E, k, p, S, \fr}$)}{}
\If{$\G$ uses one parity index $d$}
	{\If{$d$ is odd}
	 {Return \Call{Frugal-Safety}{$\S = \zug{V, E, k, S, \fr}$}\;}
	\Else
	 {Return \Call{Frugal-Reachability}{$\R=\zug{V, E, k, S, \fr}$}\;}
}

 Define $E' = \set{\zug{v, v'} \in E: v \in \notFd}$ and $E'' = \set{\zug{u, u'} \in E: u \in F_d}$.\;
 Define $\G_0 = \zug{V, E', k, p_{|\notFd}, S \cup F_d, \fr_{\G_0}}$ with $\fr_{\G_0}(u) = \begin{cases}\fr(u) & \text{ if } u \in S\\ k+1 & \text{ if } u \in F_d\end{cases}$.\; 
 $\thresh_{\G_0} = \Call{Frugal-Parity-Threshold}{\G_0}$\;
 Define $\R_0 = \zug{V, E'', k, (\notFd) \cup S, \fr_{\R_i}}$.\;
 $\thresh_{\R_0} = \Call{Frugal-Reachability-Threshold}{\R_0}$\;
\For{$i=1,\ldots$}
{
  $\thresh_{\G_i} = \Call{Frugal-Parity-Threshold}{\G_i = \zug{V, E', k, p', S \cup F_d, \fr_{\G_i} = \fr \cup \thresh_{\R_{i-1}}}}$\;
 $\thresh_{\R_i} = \Call{Frugal-Reachability-Threshold}{\R_i = \zug{V, E'', k,  \notFd, \thresh_{\G_i}}}$\;
 For each $v \in F_d$, define $\fr_{i+1}(v) = \thresh_{\R_i}(v)$\;
\If {$\fr_i(v) = \fr_{i+1}(v)$, for all $v \in F_d$} 
{ Define $\thresh_{\G}(v) = \fr_i(v) \text{ for } v \in F_d$.\;
 Define $\thresh_{\G}(u) = \thresh_{\G_i}(u) \text{ for } u \in V \setminus F_d$.\;
 Return $\thresh_{\G}$\;
}
}
 \end{algorithm}

Note that since in a frugal-reachability game both the thresholds for the reachability and safety player satisfy the average property (Theorem~\ref{thm:frugal-reach}) and the algorithm boils down to repeated calls to a solution of a frugal-reachability game, it outputs a function that satisfies the average property. 

\begin{theorem}
\label{thm:fixed-point-parity}
Given a frugal-parity bidding game $\G$ with maximal index $d$, Alg.~\ref{alg:Frugal-parity} outputs the thresholds $\thresh_\G$. Moreover, $\thresh_\G$ satisfies the average property and Alg.~\ref{alg:Frugal-parity} runs in time $O\big((|V| \cdot k)^d\big)$. 
\end{theorem}

\begin{remark}
We point out that while we develop Alg.~\ref{alg:Frugal-parity} for discrete-bidding games, it can be seen as a general ``recipe'' for extending a solution for frugal-reachability games to parity bidding games. 
While the algorithm that arises from this recipe might not be optimal complexity wise, it does provide a first upper bound, and importantly, it extends a proof that thresholds in frugal-reachability games have the average property to parity bidding games. 
\end{remark}

\section{Finding threshold budgets is in \NP and co\NP}
\label{sec:NP and coNP}

We formalize the problem of finding threshold budgets as a decision problem: 
\begin{problem}
\label{prob:decision-threshold}
{\bf (Finding Threshold Budgets).}
Given a frugal-parity bidding game $\G = \zug{V, E, k, p, S, \fr}$, a vertex $v \in V$, and $\ell \in [k]$, decide whether $\thresh_\G(v) \geq \ell$.
\end{problem}

We will show that Prob.~\ref{prob:decision-threshold} is in \NP and co\NP. 
Note that a function \(T : V \rightarrow [k] \cup \{k+1\}\) can be represented using $O(|V| \cdot \log(k))$ bits, thus it is polynomial in the size of the input to Prob.~\ref{prob:decision-threshold}. 
We describe a first attempt to show membership in \NP and co\NP. Guess $T$, verify that it satisfies the average property, and accept $\zug{G, v, \ell}$ iff $T(v) \geq \ell$. 
Unfortunately, such an attempt fails. Even though by Theorem~\ref{thm:fixed-point-parity}, the thresholds satisfy the average property, Theorem~\ref{thm:non-unique} shows that there can be other functions that satisfy it. That is, it could also be the case that $T$ satisfies the average property and $T \not \equiv \thresh_\G$. 
We point out that in continuous-bidding games, if guessing such $T$ would be possible, this scheme would have succeeded since there is a unique function that satisfies the continuous average property (Theorem~\ref{thm:cont-reach}). 

In the remainder of this section, we will show that the following problem is in \NP and co\NP by reducing it to solving a turn-based parity game of size linear in the size of the graph (and not in the size of the encoding of \(k\)). 
Then an algorithm for Prob.~\ref{prob:decision-threshold} guesses both \(T\) and winning strategies in the turn-based game. 

\begin{problem}
\label{prob:thresh-equiv}
{\bf (Verifying a guess of $T$).} Given a frugal-parity discrete-bidding game $\G$ with vertices $V$ and a function \(T : V \rightarrow [k] \cup \{k+1\}\) that satisfies the average property, decide whether \(T \equiv \thresh_\G\).
\end{problem}

We describe the high-level idea. We find it instrumental to first recall an \NP algorithm to decide whether \PO wins a turn-based parity game from an initial vertex $v_0$. 
The algorithm first guesses a memoryless strategy $f$, which is a function that maps each vertex $v$ that is controlled by \PO to an outgoing edge from $v$. The algorithm then verifies that $f$ is winning for \PO. 
The verification proceeds as follows. We solve the following problem: given a \PO strategy $f$, check whether \PT has a counter strategy $g$ such that $\play(v_0, f, g)$ violates \PO's objective. If we find that \PT has a counter strategy $g$ to $f$, then $f$ is {\em not} winning, and we reject the guess. On the other hand, if \PT cannot counter $f$, then $f$ is winning, and we accept the guess. Deciding whether \PT can counter $f$ is done as follows. We trim every edge in the game that does not comply with $f$. This leaves a graph with only \PT choices, and we check if in every reachable SCC the highest priority index is odd. 
There exists a reachable SCC where the highest priority is even iff \PT can counter $f$ iff $f$ is {\em not} winning.\footnote{An alternative description of the verification algorithm is the following. View the trimmed graph as an automaton with a singleton alphabet whose acceptance condition is \PT's objective, and check whether the language of the automaton is empty. The language is empty iff \PT cannot counter $f$.}

Our algorithm for frugal-parity games follows conceptually similar steps. Let \(T : V \rightarrow [k] \cup \{k+1\}\) that satisfies the average property. 
We verify whether \(T \equiv \thresh_\G\) as follows. We construct a partial strategy $f_T$ based on $T$. Recall that a partial strategy proposes a bid and a set of allowed vertices in each configuration. That is, guessing $T$ is not quite like guessing a strategy as in the algorithm above, rather $T$ gives rise to a partial strategy. 
We seek a \PO strategy that agrees with $f_T$ and wins when the game starting from every configuration $\zug{v, T(v)}$ with $T(v) < k+1$. 
Note that if $T \equiv \thresh_\G$, then such a strategy exists. 
Given $f_T$, we describe an algorithm that decides whether \PT can counter every \PO strategy that agrees with $f_T$. Our algorithm constructs and solves a turn-based parity game.

\subsection{From bidding games to turn-based games}
Let  \(T\) be a function that satisfies the average property. Recall that the partial strategy \(f_T\) that is constructed in Sec.~\ref{sec:def-f_T} is a function that, given a configuration $\zug{v, B}$, outputs $\zug{b, A}$, where $b \leq B$ is a bid and $A \subseteq V$ is a subset of neighbors of $v$ that are called {\em allowed vertices}. A strategy $f'$ agrees with $f_T$ if from each configuration, it bids the same as $f_T$ and chooses an allowed vertex upon winning the bidding. 

We construct a parity turn-based game $G_{T, \G}$ such that, roughly, if \PO wins in every vertex in $G_{T, \G}$, then \PO has a strategy $f'$ that agrees with $f_T$ and wins from every configuration $\zug{v, T(v)}$ in $\G$, thus $T \geq \thresh_\G$. 

We describe the intuition behind the construction of $G_{T, \G}$. Consider the following first attempt to construct $G_{T, \G}$. 
Recall the construction in Sec.~\ref{sec:bidding2conc} of the explicit concurrent game that corresponds to $\G$, and denote it by $\G'$. The vertices of $\G'$ are the configurations $\C$ of $\G$. We construct a game $\G''$ on the configuration graph $\C$. 
Recall that our goal is to check whether \PT can counter \PO's strategy, which can be thought of as \PT responds to \PO's actions in each turn. Thus, $\G''$ is turn-based: when the game is in configuration $c$, \PO first chooses $\zug{b_1, v_1}$, and only then, \PT responds by choosing an action $\zug{b_2, v_2}$. The next configuration is determined by these two actions in the same manner as the concurrent game. 
Next, we trim \PO actions in $\G'$ that do not comply with $f_T$: in a configuration $c = \zug{v, B}$ in $\G''$ with $\zug{b, A} = f_T(c)$, \PO must bid $b$ and choose a vertex in $A$. That is, an action $\zug{b', v'}$ is not allowed if $b' \neq b$ or if $v' \notin A$. Finally, we omit \PT actions that are dominated: observing a \PO bid of $b$, she chooses between bidding $0$ and letting \PO win the bidding or bidding $\succb{b}$ and winning the bidding. 
It is not hard to see that \PO wins $\G''$ from configuration $\zug{v, B}$ iff there is a strategy $f'$ that agrees with $f_T$ and wins $\G$ from $\zug{v, B}$.

The first attempt fails since the size of $\G''$ is proportional to the number of configurations, which is exponential in $\G$. 
We overcome this key challenge as follows. 
Lemma~\ref{lem: invariant} shows that when $\G$ starts from configuration $\zug{v, B}$ with $B \geq T(v)$ a strategy $f'$ that agrees with $f_T$ maintains an invariant on \PO's budget: the game only reaches configurations of the form $\zug{v', B'}$ with $B' \geq T(v')$. 
We shrink the size of the game by grouping all configurations in which \PO's budget is greater than $\succb{T(v)}$ into a vertex denoted $\zug{v, \top}$. 

We describe the idea that allows keeping only three copies of each vertex (see details in Lemma~\ref{lemma:TBparity-iff-thresh}). 
We refer to the distance from the invariant as {\em spare change}, formally at $v$ with budget $B \geq T(v)$, the spare change is \(|B \ominus T(v)|\). 
Recall from Sec.~\ref{sec:def-f_T} that $f_T$ chooses one of two bids in a vertex $v \in V$, and the choice depends on the advantage status and does not depend on the spare change. Thus, our winning strategy in $\G$ emulates a winning strategy in $G_{T,\G}$: both bid according to $f_T$ and the latter prescribes a vertex to move to upon winning a bidding. Thus, a play $\pi$ in $\G$ corresponds to a play $\playtg$ in $G_{T,\G}$.
There can be three outcomes: (1) $\playtg$ is infinite, (2) $\playtg$ ends in a sink $S$, or (3) $\playtg$ ends in a sink $\zug{v,\top}$. The first two cases mean that $\pi$ is winning in $\G$. When Case~(3) occurs and $G_{T, \G}$ reaches $\zug{v, \top}$, then $\G$ reaches $\zug{v, B}$ with $B > \succb{T(v)}$, and we restart $G_{T, \G}$ from either $\zug{v, T(v)}$ or $\zug{v, \succb{T(v)}}$ depending on the advantage status. Note that, after restarting the game, \PO plays the same except that his spare change increased. 
This is the key idea. Since whenever Case~(3) occurs, \PO's spare change strictly increases and the spare change cannot exceed the total budget $k$, Case~(3) can occur only finitely often.

Formally, we define the turn-based parity game $G_{T, \G} = \zug{V_1, V_2, E, p}$. The vertices controlled by \PLi are $V_i$, for $i \in \set{1,2}$, where
	\begin{itemize}
		\item $V_1 = \{\zug{v,T(v)}, \zug{v, \succb{T(v)}}, \zug{v, \top}: v \in (V \cup S)\}$ and 
		\item $V_2 = \set{\zug{v, c}: v \in V, c \in \C}$.
	\end{itemize}
    Note that it is possible that $T(v) = k+1$, which as in  Remark~\ref{rem:losing-vertices}, means that \PO loses from $v$ in $\G$ with every initial budget. 
    We define each vertex $\zug{v, k+1}$ as losing for \PO in $G_{T, \G}$: it is a sink with even parity index. 
	We define the edges in the game. A vertex $\zug{v, B}$ is a sink if $v \in S$ or if $B = \top$. Consider $c = \zug{v, B} \in V_1$ and let $\zug{b, A} = f_T(c)$. The neighbors of $c$ are $\set{\zug{v', c}: v' \in A}$. Intuitively, $\zug{v', c}$ means that \PO chooses the action $\zug{b,v'}$ at configuration $c$; the bid $b$ is determined by $f_T$ and $v'$ is an allowed vertex. A vertex $\zug{v', c}$ is a \PT vertex. Intuitively, \PT makes two choices: who wins the bidding and where the token moves upon winning. Thus, a vertex $\zug{v', c}$ has two types of neighbors, depending on who wins the bidding at $c$:
	\begin{itemize}
		\item First, $\zug{v', B \ominus b}$ is a neighbor of $\zug{v', c}$, meaning \PT lets \PO win the bidding by bidding $0$.
		\item Second, suppose that $k^* \ominus B \geq \succb{b}$, i.e., \PT has sufficient budget to win the bidding. Let $B' = B \oplus (\succb{b})$ be \PO's updated budget and $w \in N(v)$. If $c' = \zug{w, B'} \in V_1$, then $c'$ is a neighbor of $\zug{v', c}$. We note that $c' \notin V_1$ when $B'$ exceeds $\succb{T(w)}$, then we trim the budget and set  $\zug{w, \top}$ as a neighbor of $\zug{v', c}$.
	\end{itemize}
	For ease of presentation, we define parity indices only in \PO vertices. 
	A non-sink vertex in $G_{T, \G}$ ``inherits'' its parity index from the vertex in $\G$; namely, for \(c = \zug{v,B} \in V_1\), we define $p'(c) = p(v)$.  The parity index of a sink is odd so that \PO wins in sinks.

	\begin{figure}[ht]
	\captionsetup{justification=centering}
	\begin{subfigure}{0.48\linewidth}
          \begin{center}
                    	\includegraphics[width=0.9\textwidth]{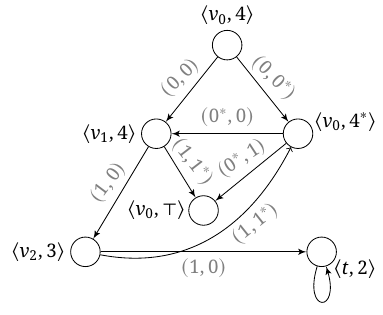}
                      \end{center}
			

			

			
		\caption{The turn-based game $G_{T_1, \G_1}$. \PO loses from some vertices, thus $T_1 \not \equiv \thresh_{\G_1}$}\label{fig:turn-based-T1}
	\end{subfigure}\hfill
	\begin{subfigure}{0.48\textwidth}
          \captionsetup{justification=centering}
          \begin{center}
                             	\includegraphics[width=0.95\textwidth]{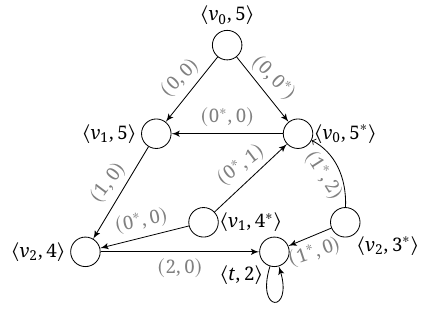}
                              \end{center}
		\caption{The turn-based game $G_{T_2, \G_1}$. \PO wins from all vertices, thus $T_2 \equiv \thresh_{\G_1}$}\label{fig:turn-based-T2}
	\end{subfigure}
	\caption{Turn-based reachability games that correspond to \(\G_1\) from Fig.~\ref{fig:avgbutnotthreshold} for two functions that satisfy the average property. \PO vertices are omitted; that is, all depicted vertices are \PT vertices. The edge labeling depict bidding outcomes and are meant to ease presentation.}
	\label{fig:turn-based}
\end{figure}

\begin{example}
Fig.~\ref{fig:avgbutnotthreshold} depicts a frugal-reachability bidding game \(\G_1\) with two functions that satisfy the average property: \(T_1 = \zug{4,3^*, 3, 2}\) and \(T_2 = \zug{5,4^*, 3^*, 2}\). 
Fig.~\ref{fig:turn-based} depicts the games $G_{T_1, \G_1}$ and $G_{T_2, \G_1}$. For ease of presentation, Fig.~\ref{fig:turn-based} is slightly inconsistent with the construction of the games. 
The reason is that both $f_{T_1}$ and $f_{T_2}$ prescribe a singleton set of allowed vertices from all configurations, thus \PO makes no choices in the game. We thus skip his vertices and simplify \PT's vertices: each vertex in Fig.~\ref{fig:turn-based} corresponds to a configuration, and all vertices are controlled by \PT. \PO's goal in both games is to reach a sink.  An outgoing edge from vertex $c$ labeled by $\zug{b_1, b_2}$ represents the outcome of a bidding at configuration $c$ in which \PLi bids $b_i$, for $i \in \set{1,2}$. Thus, each vertex $c$ has two outgoing edges labeled by $\zug{b_1, 0}$ and $\zug{b_1, \succb{b_1}}$, where $b_1$ is the bid that $f_{T_1}$ or $f_{T_2}$ prescribes at $c$. Note that some edges are disallowed. For example, in the configuration $\zug{v_1, 5}$ in $G_{T_2, \G_1}$, the bid prescribed by $f_{T_2}$ is $b_1 = 1$ and \PT cannot bid $\succb{b_1} = 1^*$ since it exceeds her available budget (indeed, $k=5$, thus \PT's budget in $c$ is $k^* \ominus 5 = 0^*$). 

Note that \(G_{T_1, \G_1}\) has a cycle. Thus, \PO does not win from every vertex and $T_1$ does not coincide with the threshold budgets. On the other hand, \(G_{T_2, \G_1}\) is a DAG. Thus, no matter how \PT plays, \PO wins from all vertices, and indeed \(T_2 \equiv \thresh_{\G_1}\).
\end{example}

\subsection{Correctness}
In this section, we prove soundness and completeness of the approach. We start with soundness.

\begin{restatable}{lemma}{TBparityiffthresh}
	\label{lemma:TBparity-iff-thresh}
	If Player 1 wins from every vertex $\zug{v,B}$ in $G_{T, \G}$ with $B < k+1$, then $T \geq \thresh_\G$. 
\end{restatable}

\begin{proof}
	Suppose that \PO wins from every such vertex of \(G_{T, \G}\) and let \(\ftg\) be a \PO memoryless winning strategy. 
	We construct a strategy \(f\) in \(\G\) based on \(\ftg\) and show that it is winning from every configuration \(\zug{v, B}\) where \(B \geq T(v)\). 
	This implies that \(T \geq \thresh_{\G}\) since \(f\) witnesses that \PO can win with a budget of \(T(v)\) from \(v\). Note that we do not yet rule out that a different strategy wins with a lower budget, this will come later. 
	
	We introduce notation. 
	Consider a configuration \(c = \zug{v, B}\) in \(\G\) with $T(v) < k+1$ and \(B \geq T(v)\). 
	The vertex in \(G_{T, \G}\) that \emph{agrees} with \(c\), denoted by \(\configtg\), is the vertex in \(\{\zug{v, T(v)}\, \zug{v, \succb{T(v)}}\}\) that matches with \(c\) on the status of the advantage (and of course on the vertex of \(\G\)). Note the convention of calling $c$ a {\em configuration} in $\G$ and a {\em vertex} in $G_{T,\G}$. 
	For example if \(T(v) = 5^*\) for some vertex \(v\) and \(c = \zug{v, 9}\) is a configuration in \(\G\), then the vertex of \(G_{T, \G}\) that \emph{agrees} with \(c\), denoted by \(\configtg\), is \(\zug{v, 6}\).  
	Recall that even though the budget in \(c\) may be higher than that of \(\configtg\), the partial strategy \(f_T\) acts the same in both, i.e., \(f_T(c) = f_T(\configtg)\). 
	The \emph{spare change} that is associated with \(c\), denoted by \(\Spare(c)\) is \(|B| - |T(v)|\). 
	
	In the following, we construct \(f\) based on \(f_T\) and \(\ftg\). 
	Specifically, we define \(f\) to agree with \(f_T\) on the bid and choose the successor vertex according to \(\ftg\).
	Let \(\zug{b, A} = f_T(c)\).  
	Recall that \(\configtg\) is a \PO vertex in \(G_{T, \G}\) and its neighbours are of the form \(\zug{v', c}\) such that \(v'\) is an allowed vertex, i.e, \(v' \in A\). 
	Intuitively, proceeding to vertex \(\zug{v', \configtg}\) in \(G_{T, \G}\) is associated with moving to \(v'\) upon winning the bidding at \(\configtg\).
	Let \(\zug{v', \configtg} = \ftg(\configtg)\). 
	Then, we define \(f(c) = \zug{b, v'}\). 
	
	We claim that \(f\) is winning from an initial configuration \(c_0 = \zug{v, B}\) in \(\G\) with $T(v) < k+1$ and \(B \geq T(v)\). 
	Let \(g\) be a \PT strategy in \(\G\). 
	The initial vertex \(\configtg{c}{0}\) in \(G_{T, \G}\) is the vertex that \emph{agrees} with \(c_0\).
	We construct a \PT strategy \(\ftg{g}\) in \(G_{T, \G}\) so that \(\play(\configtg{c}{0}, \ftg, \ftg{g})\) in \(G_{T, \G}\) simulates \(\play(c_0, f, g)\) in \(\G\): when \(\play(c_0, f, g)\) is in a configuration \(c\), \(\play(\configtg{c}{0}, \ftg, \ftg{g})\) is in a vertex \(\configtg{c}\) that agrees with \(c\). 
	
	We define \(\ftg{g}\) inductively. 
Initially the invariant holds due to our choice of \(\configtg{c}{0}\) in \(G_{T, \G}\). 
	Suppose that \(\G\) is at configuration \(c = \zug{v, B}\), then the \(\play(\configtg{c}{0}, \ftg, \ftg{g})\) in  \(G_{T, \G}\) is at the vertex \(c^*\) that agrees with \(c\). 
	Denote by \(\placeholdr \in \{T(v), \succb{T(v)}\}\) such that \(c^* = \zug{v, \placeholdr}\). 
	Let \(\zug{b_1, v_1} = f(c)\) as defined above, let \(\zug{b_2, v_2} = g(c)\) be \PT's choice, and let \(d\) be the next configuration in $\G$. We extend the play in $G_{T,\G}$ as follows. 
	We first register \PO's move in \(G_{T, \G}\) by proceeding to the \PT vertex \(\zug{v_1, c^*}\). 
	We distinguish between two cases. First, \PO wins the bidding in $\G$. We define \(\ftg{g}\) to choose \(\zug{v_1, \placeholdr \ominus b_1}\) as the successor vertex from \(\zug{v_1, c^*}\). 
	Note that, in this case, the configuration \(d\) is \(d^* = \zug{v_1, B \ominus b_1}\).  
	Since \(c^*\) agrees with \(c\), then \(d^*\) agrees with \(d\). 
	Second, \PT wins the bidding in $\G$. We define \(\ftg{g}\) to proceed to \(d^* = \zug{v_2, \placeholdr \oplus b_2}\), if it exists in \(G_{T, \G}\). 
	If \(d^*\) is in the graph, then again, $d^*$ agrees with $d$. 
	On the other hand, if \(d^*\) is not a vertex in \(G_{T, \G}\), it intuitively means \(\placeholdr \oplus b_2 > \succb{T(v_2)}\), and we define \(\ftg{g}\) to proceed to \(\zug{v_2, \top}\). 
	Let us denote \(\hat{d^*}\) be the vertex in \(G_{T, \G}\) that agrees with \(d\). 
	We apply the same definition above starting from vertex \(\hat{d^*}\) in \(G_{T, \G}\). 
	That is, in the next turn, assuming that \PO proceeds in \(G_{T, \G}\) to \(\zug{v_2, \hat{d^*}}\), we define \(\ftg{g}\) according to \(g(\zug{v_2, B \oplus b_2})\), as discussed above. 
	We call this a \emph{restart} in the simulation. 
	Note that if the simulation is not a restart, then \(\Spare(c) = \Spare(d)\), and if it is a restart, then \(\Spare(c) < \Spare(d)\). 
	
	We claim that $\play(c_0, f, g) = c_0, c_1,\ldots$ is winning for \PO in $\G$. We slightly abuse notation and denote by $\play(\configtg{c}{0}, \ftg, \ftg{g}) = \configtg{c}{0}, c^*_1,\ldots$ the sequence of \PO vertices that are traversed in $G_{T, \G}$, that is we skip \PT vertices. 
	Since $\ftg$ is winning in $G_{T, \G}$, then $\play(\configtg{c}{0}, \ftg, \ftg{g})$ is winning for \PO. We distinguish between three cases. 
	First, $\play(\configtg{c}{0}, \ftg, \ftg{g})$ is infinite. Since for every $i \geq 0$, the vertex $c^*_i$ agrees with $c_i$, the two plays agree on the parity indices that are visited, thus $\play(c_0, f, g)$ satisfies the parity objective. 
	Second, $\play(\configtg{c}{0}, \ftg, \ftg{g})$ is finite and ends in a sink configuration $c_k = \zug{s, B}$. Note that $B \geq T(s)$, and the definition of $T$ requires that $T(s) = \fr(s)$. 
	Since $c^*_k$ agrees with $c_k$, it follows that $c_k$ satisfies the frugal objective. 
	Third, $\play(\configtg{c}{0}, \ftg, \ftg{g})$ is finite and ends in $c^*_k = \zug{v, \top}$. 
	Let $\hat{c^*_k}$ denote the vertex that agrees with $c_k$. 
	We apply the reasoning above to $\play(\hat{c^*_k}, \ftg, \ftg{g})$. 
	
	Note that the third case can occur only finitely many times, since a restart causes the spare change to strictly increase, and the spare change is bounded by $k$. Thus, eventually, the play in $G_{T, \G}$ falls into one of the first two cases, which implies that $\play(c_0, f, g)$ is winning for \PO. 	
\end{proof}

\begin{remark}
The proof of Lemma~\ref{lemma:TBparity-iff-thresh} constructs a \PO winning strategy $f$ in $\G$. Note that in order to implement $f$, we only need to keep track of a vertex in $G_{T, \G}$. Thus, its memory size equals the size of $G_{T, \G}$, which is linear in the size of $\G$. This is significantly smaller than previously known constructions in parity and reachability discrete-bidding games, where the strategy size is polynomial in $k$, and is thus exponential when $k$ is given in binary.
\end{remark}

The following lemma shows completeness; namely, that a correct guess of $T$ implies that \PO wins from every vertex in $G_{T, \G}$. 

\begin{restatable}{lemma}{thresholdpassesthetest}
	\label{lemma:thresholdpassesthetest}
	If $T \equiv \thresh_\G$, then \PO wins from every vertex $\zug{v, B}$ in $G_{T,\G}$ with $B \in \Natstr$ and $B < k+1$. 
\end{restatable}

\begin{proof}
Assume towards a contradiction that $T  \equiv \thresh_{\G}$ and there is a \PO vertex \(\configtg{c}{0} = \zug{v,B}\) in \(G_{T, \G}\) that is losing for \PO. 
Let \(\ftg{g}\) be a \PT memoryless strategy that wins from vertex $\configtg{c}{0}$ in \(G_{T, \G}\). 
Recall that \(B \in \{T(v), \succb{T(v)}\}\). 
Note that $B \geq T(v)$. Since we assume $T \equiv \thresh_\G$, and \(B < k+1\) \PO wins from configuration \(c_0 = \zug{v, B}\) in $\G$. Let \(f\) be a \PO winning strategy from \(c_0\) in \(\G\). Note that we follow the convention of referring to $c$ in $\G$ as a {\em configuration} and \(\configtg\) in $G_{T,\G}$ as a {\em vertex}, even though both are \(\zug{v, B}\). 
We will reach a contradiction by constructing a \PT strategy $g$ in $\G$ that counters $f$, thus showing that $f$ is not winning. 
Recall that a winning \PO strategy can be thought of as a strategy that, in each turn, reveals \PO's action first, and allows \PT to respond to \PO's action. 

We construct a \PO strategy $\ftg$ in $G_{T,\G}$ based on $f$ as long as $f$ agrees with $f_T$ and a \PT strategy $g$ in $\G$ based on $\ftg{g}$. Both constructions are straightforward. First, for $\ftg$, consider a \PO vertex $\configtg$ in $G_{T,\G}$. Recall that $\configtg$ is a configuration in $\G$, which we denote by $c$ to avoid confusion. Suppose that $f(c)$ agrees with $f_T(c)$, that is denoting $\zug{b, A} = f_T(c)$, we have $\zug{b, v} = f(c)$ with $v \in A$. 
Then, in $G_{T,\G}$, from vertex $\configtg$, the strategy $\ftg$ proceeds to $\zug{v, \configtg}$. We stress that $\ftg$ is only defined when $f$ agrees with $f_T$. 
Second, for $g$, recall that \PT vertices in $G_{T, \G}$ are of the form $\zug{v, c}$, and \PT chooses between, intuitively letting \PO win the bidding or bidding $\succb{b}$, winning the bidding, and choosing the next vertex. Assume $\zug{b, v} = f(c)$ that agrees with $f_T$, then $g$ responds by following $\ftg{g}$: if $\ftg{g}$ lets \PO win from $\zug{v, \configtg}$, then $g$ bids $0$ in $c$ and lets \PO win the bidding, and if it wins the bidding by proceeding to vertex $\zug{v', B'}$, then $g$ chooses $\zug{\succb{b}, v'}$, i.e., it too wins the bidding in $c$ and proceeds to vertex $v'$.

Let $\playtg$ and $\pi$ respectively denote the longest histories of $G_{T, \G}$ and $\G$ that start from $\configtg{c}{0}$ and \(c_0\) that arise from applying $\ftg$ against \(\ftg{g}\) in \(G_{T, \G}\) and \(f\) against \(g\) in \(\G\), as long as $f$ agrees with $f_T$. Note that, skipping \PT vertices in $G_{T,\G}$, the plays $\playtg$ and $\pi$ traverse the same sequence of configurations. We claim that the two plays cannot be infinite. Indeed, assume otherwise, then since we assume $f$ is winning, $\pi$ satisfies \PO's objective, and since we assume $\ftg{g}$ is winning, $\playtg$ violates \PO's objective, but both cannot hold at the same time. 
Also, $\playtg$ cannot end in a sink. Indeed, visiting a sink that is winning for \PO is not possible since $\playtg$ is consistent with a \PT winning strategy, and visiting a winning sink $\zug{v, k+1}$ for \PT is not possible since it means that $\pi$, which is consistent with a \PO winning strategy in $\G$, visits a losing configuration for \PO.   We conclude that $\pi$ and $\playtg$ are finite and end in a configuration in which \(f\) does not agree with~\(f_T\). 

Let $c = \zug{v, B}$, where \(B \in \{T(v), \succb{T(v)}\}\), be the last configuration in $\pi$. 
That is, $c$ is the first configuration in which $f$ chooses an action that does not agree with $f_T$. Let $\zug{b,A} = f_T(c)$ and $\zug{b_1, v_1} = f(\pi)$. 
In the remainder of the proof, we consider the three ways in which $f$ can disagree with $f_T$. 
In each of these cases, we subsequently define \PT response \(g\), and show that she can win from the resulting configuration. 

In particular, we show that in all but one subcases of these three cases, we have a ``suitable'' \PT response by \(g\) which results in a configuration of the form \(c = \zug{v', B'}\), where \(B' < T(v')\). 
Thus, the standard argument follows (in all but one subcase) from there as: because the budget ``falls'' below the threshold budget (by hypothesis, \(T \equiv \thresh\)), by definition, \PT has a winning strategy in \(\G\) from there onwards, and \(g\) simply follows that. 
In the remaining subcase (Case 2. (ii), in particular), we will see that even though this is also a way how \(f\) differs from \(f_T\), it does not necessarily result \PO's budget falling below \(T(v')\) (assuming the resulting vertex is \(v\)) for any \PT response. 
But in this case, we argue that \(f\) eventually differs from \(f_T\) by other means. 
Hence, even though we may not have the desired ``suitability'' in  \PT's response \(g\) in this case, we will eventually encounter it when \(f\) eventually differs from \(f_T\) by other means.

We recall Observation~\ref{obs: TmatcheswithBid}, which intuitively states that when \PO has the advantage and the bids of $f$ and $f_T$ agree, then \PO uses the advantage.

We finally proceed to analyze the three ways in which $f$ disagrees with $f_T$:

\noindent{\bf Case 1: $f$ underbids; $b_1 < b$.} 
\PT responds by bidding $b \ominus 0^*$. 
We show that she wins the bidding, but before that 
we show  \PT can indeed bid \(b \ominus 0^*\) at vertex \(v\) from her budget \(k^* \ominus B\).
Note that, here \(b = \tbid{v}\) if \(B = T(v)\), and \(b = \succb{\tbid{v}}\), if \(B = \succb{T(v)}\). 

\begin{claim*}\label{clm: PTcanbid}
	When \PO has a budget \(T(v)\) (alternatively, \(\succb{T(v)}\)), \PT can bid \(\predb{\tbid{v}}\) (\(\tbid{v}\) respectively). 
\end{claim*}

\begin{subproof}
	We analyze the case when  \(B = T(v)\), as the other case is exactly similar.
	So, when \PO's budget is \(B = T(v)\), \PT's budget is \(k^* \ominus T(v)\). 
	In order to establish that \PT can indeed bid \(\predb{b}\), we show the following:
	\begin{equation*}
		(k^* \ominus T(v)) \ominus (\predb{\tbid{v}}) \geq 0
	\end{equation*}
	
	We prove this by a case analysis akin to the proof of Lemma~\ref{lem:resultonT}, i.e, we analyze four cases, each of which corresponds to a parity of \(\sumT\) and an advantage status of \(T(\vminus)\). 
	\begin{itemize}
		\item \(\sumT\) is even and \(T(\vminus) \in \Nat\). 
		
		In this case, 
		\begin{align*}
			\adjustbrac{k^* \ominus T(v)} \ominus \adjustbrac{\predb{\tbid{v}}} &= \adjustbrac{k^* \ominus \half{\sumT}} \ominus \succInt{\half{\diffT} -1}\\
			&= \succInt{k - \half{\sumT}} \ominus \succInt{\half{\diffT} -1} \\
			&= k - \half{\sumT} - \half{\diffT} + 1\\
			&= \adjustbrac{k+1} - \absolut{T(\vplus)} \geq 0 
		\end{align*}
	
	\item \(\sumT\) is odd and \(T(\vminus) \in \Natstro\). 
	
	In this case, we have 
	\begin{align*}
		\adjustbrac{k^* \ominus T(v)} \ominus (\predb{\tbid{v}}) &= \adjustbrac{k^* \ominus \adjustbrac{\floor{\half{\sumT}}+1}} \ominus \succInt{\floor{\half{\diffT}} -1}\\
		&= k - \floor{\half{\sumT}} - 1 - \floor{\half{\diffT}} + 1\\
		&= k+1 - \adjustbrac{\half{\sumT} + \half{\diffT}}\\
		&= \adjustbrac{k+1} - \absolut{T(\vplus)} \geq 0
	\end{align*}

	\item \(\sumT\) is even and \(T(\vminus) \in \Natstro\). 
	
	In this case, it goes as following:
	\begin{align*}
		\adjustbrac{k^* \ominus T(v)} \ominus (\predb{\tbid{v}}) &=
		\adjustbrac{k^* \ominus \succInt{\half{\sumT}}} \ominus \adjustbrac{\half{\diffT} -1}\\
		&= \adjustbrac{k+1} - T(\vplus) \geq 0
	\end{align*}

	\item Finally, \(\sumT\) is odd and \(T(\vminus) \in \Nat\). 
	
	Here, we have
	\begin{align*}
		\adjustbrac{k^* \ominus T(v)} \ominus (\predb{\tbid{v}}) &=
		\adjustbrac{k^* \ominus \succInt{\floor{\half{\sumT}}}} \ominus \floor{\half{\diffT}}\\
		&= k - \half{\sumT} + \half{1} - \half{\diffT} + \half{1}\\
		&= \adjustbrac{k+1} - T(\vplus) \geq 0
	\end{align*}
	\end{itemize}
Therefore, we conclude that when \PO has a budget of \(T(v)\), \PT does have the enough budget to bid \(\predb{\tbid{v}}\).
\end{subproof}
 
She then proceeds to a neighbour $v'$ with \(T\)-value \(T(v^+)\). Let $c' = \zug{v', B'}$ denote the resulting configuration. Intuitively, \PT pays less than she should for winning the bidding. Formally, we will show that $B' < T(v')$. This will conclude the proof. Indeed, since we assume $T \equiv \thresh_\G$, \PT has a winning strategy from $c'$, which she uses to counter $f$ from \(c'\) onwards. 
We distinguish between two cases depending on whether \PO holds the advantage:
\begin{enumerate}[(a)]
	\item \PO holds the advantage, i.e., \(B \in \Natstro\). By Observation~\ref{obs: TmatcheswithBid}, he uses it according to \(f_T\), thus \(b \in \Natstro\). 
	\PT bids $b_2 = b \ominus 0^*$. 
	First, note that the bid is legal. Indeed, since $b$ contains the advantage, $b_2$ does not. 
	Second, note that \PT wins the bidding. 
	Indeed, if \PO bids less than $b_2$, clearly \PT wins, and if he bids $b_2$, then a tie occurs, and since he has the advantage and does not use it, \PT wins the bidding. 
	As a result, \PO's budget is updated to \(B \oplus (b \ominus 0^*) < B \oplus b < B \oplus (\succb{b}) = |T(v^+)|^*\), in particular, \(B \oplus (b \ominus 0^*) < T(v^+)\). 
	
	\item \PO does not hold the advantage, i.e., \(B \in \Nat\). Again, by Observation~\ref{obs: TmatcheswithBid}, $b$ does not include the advantage and bidding \(b_1 < b\) necessarily implies \(b_1 < b \ominus 0^*\), simply because he does not hold the advantage. 
 \PT bids \(b \ominus 0^* \in \Natstro\). It is not hard to see that she wins the bidding and showing that $B' < T(v')$ is done as in the previous case. 
\end{enumerate}

\noindent{\bf Case 2: $f$ over bids; $b_1 > b$.} 
We assume \(B = T(v)\) and the case of \(B = \succb{T(v)}\) is similar. Note that Observation~\ref{obs: TmatcheswithBid} implies that $f_T$ proposes a bid of \(\succb{b}\) when \PO's budget is \(\succb{T(v)}\). 
Intuitively, if \PO wins the bidding with his bid of $b_1$, he will pay ``too much'', and \PT indeed lets him win by bidding $0$ (except for one case that we will explain later). 
The resulting configuration is \(c' = \zug{v_1, B \ominus b_1}\), and we will show that \(B \ominus b_1 < T(v_1)\) (barring one case). As in the underbidding  case, this concludes the proof: since we assume $T \equiv \thresh_\G$, \PT wins from $c'$. 

We first consider the easier case when \(b_1 > \succb{b}\). Then, \(B \ominus b_1 < B \ominus (\succb{b}) \leq T(v^-) \leq T(v_1)\), thus $B \ominus b_1 < T(v_1)$, as required. We proceed to the harder case of \(b_1 = \succb{b}\). Note that this case necessarily implicates that \PO has the advantage, i.e, \(B \in \Natstro\). 
Indeed, otherwise when he does not have the advantage, i.e., \(B \in \Nat\) then \(b \in \Nat\) too (from Observation~\ref{obs: TmatcheswithBid}), thus he cannot bid \(\succb{b}\), which is in \(\Natstro\), from his budget \(B\). 
Recall from Definition~\ref{def:average} that when \PO's budget \(B = T(v)  \in \Natstro\), there are two possibilities: (i) \(|T(v^+)| + |T(v^-)|\) is odd and \(T(v^-) \in \Nat\), and (ii) \(|T(v^+)| + |T(v^-)|\) is even and \(T(v^-) \in \Natstro\). In Case~(i), \(T(v) - b = T(v^-)\), hence when \PO bids \(\succb{b}\), \PT's response is \(0\), and \PO's budget in the next configuration is strictly lower than the threshold. 

We conclude with Case~(ii). Recall that in this case \(b = \lfloor \frac{|T(v^+)| - |T(v^-)|}{2} \rfloor \ominus 0^*\). This case requires a different approach since \PO can bid \(\succb{b}\) and even if he wins the bidding, his budget in the next configuration does not fall below the threshold. We define $g$ to follow $\ftg{g}$. Consider the move of $\ftg{g}$ from \(\zug{v_1, c}\). If it lets \PO win by proceeding to \(\zug{v_1, B \ominus b}\), then $g$ responds to \(f\) in \(\G\) by bidding \(0\). Recall that \PT's other action in $G_{T,\G}$ corresponds to a bid of $\succb{b}$, and is represented by proceeding to vertex \(\zug{v', B \ominus (\succb{b})}\). Then, in $\G$, we define $g$ to bid \(\succb{b}\), thus both players bid $\succb{b}$ and \PT wins the tie since \PO has the advantage and does not use it. \PT proceeds to \(v'\) following $\ftg{g}$. The key idea is that in both cases, we reach the same configuration in $\G$ and $G_{T,\G}$. That is, even though $f$ disagrees with $f_T$, we extend the two plays \(\pi\) and \(\playtg\) and restart the proof. As discussed in the beginning of the proof, the plays cannot be infinite, thus eventually $f$ disagrees with $f_T$ in one of the other manners. 

\noindent{\bf Case 3: $f$ does not choose an allowed vertex; $b_1 = b$ and $v_1 \notin A$.}
Recall that, by definition, the set \(A\) of allowed vertices consists of all vertices \(v'\) that satisfy \(T(v) \ominus b \geq T(v')\). 
Therefore, \PT responds to \(f\) by letting \PO win by bidding \(0\). 
In the resulting configuration, \PO's budget is strictly less than $T$, which coincides with the threshold budget, and, as in the above, \PT proceeds with a winning strategy. 
\stam{
We distinguish three cases:
\begin{itemize}
\item $f$ overbids, namely $b_1 > b$,
\item $f$ underbids, namely $b_1 < b$, and
\item $f$ does not choose an allowed vertex, namely $b_1 = b$ and $v_1 \notin A$. 
\end{itemize}
}
\end{proof} 

Finally, we verify that \(T \leq \thresh_{\G}\). We define a function \(T': V \rightarrow [k] \cup \set{k+1}\) as follows. For \(v \in V\), when $T(v) > 0$ we define $T'(v) = (k+1) \ominus T(v)$, and $T'(v) = k+1$ otherwise. Lemma~\ref{lem:flippedaverage} shows that \(T'\) satisfies the average property.
We proceed as in the previous construction only from \PT's perspective. 
We construct a partial strategy $f_{T'}$ for \PT from $T'$ just as $f_T$ is constructed from $T$, and construct a turn-based parity game $G_{T', \G}$. Let \(\thresh_{\G}^2\) denote \PT's threshold function in \(\G\). That is, at a vertex $v \in V$, \PT wins when her budget is at least $\thresh_\G^2(v)$ and she loses when her budget is at most $\predb{\thresh_\G^2(v)}$. 
Applying Lemmas~\ref{lemma:TBparity-iff-thresh} and~\ref{lemma:thresholdpassesthetest} to \PT, we obtain the following.

\begin{lemma}
\label{lem:parity-PT}
    If \PT wins from every vertex $\zug{v, B}$ in $G_{T', \G}$ with $B \in \Natstr$ and $B < k+1$, then $T' \geq \thresh^2_{\G}$.
    If $T' \equiv \thresh^2_{\G}$, then \PT wins from every vertex $\zug{v, B}$ of $G_{T', \G}$ with $B \in \Natstr$ and $B < k+1$. 
\end{lemma}

Given a frugal-parity discrete-bidding game $\G = \zug{V, E, k, p, S, \fr}$, a vertex $v \in V$, and $\ell \in [k]$, we guess $T: V \rightarrow [k] \cup \set{k+1}$ and verify that it satisfies the average property. Note that the size of $T$ is polynomial in $\G$ since it consists of $|V|$ numbers each of size $O(\log k)$. We construct $G_{T, \G}$ and $G_{T', \G}$, guess memoryless strategies for \PO and \PT, respectively, and verify in polynomial time that they are indeed winning.  Finally, we check whether $T(v) \geq \ell$, and answer accordingly. Correctness follows from Lemmas~\ref{lemma:TBparity-iff-thresh},~\ref{lemma:thresholdpassesthetest}, and~\ref{lem:parity-PT}. We thus obtain our main result.

\begin{restatable}{theorem}{parityNPcoNP}
	\label{thm:parity-NP-coNP}
	The problem of finding threshold budgets in frugal-parity discrete-bidding games is in \NP and co\NP. 
\end{restatable}

\section{Discussion}
We develop two algorithms to find threshold budgets in discrete-bidding games. Our first algorithm shows, for the first time, that thresholds in parity discrete-bidding games satisfy the average property. Previously, only thresholds in reachability discrete-bidding games were known to have this property. 
We study, for the first time, the problem of computing threshold budgets in discrete-bidding games in which the budgets are given in binary, and establish membership in NP and coNP for reachability and parity objectives. 
Previous algorithms for reachability and parity discrete-bidding games have exponential running time in this setting. 
We develop novel building blocks as part of our algorithms, which can be of independent interest. First, we define and study, for the first {\em frugal} objectives, which are reachability objectives accompanied by an enforcement on a player's budget when reaching the target. Second, our fixed-point algorithm provides a recipe for extending a proof on the structure of thresholds in reachability bidding games to parity bidding games. Third, we develop, for the first time, strategies that can be implemented with linear memory in reachability and parity discrete-bidding games, whereas previous constructions used exponential memory.

We point to the intriguing state of affairs in parity discrete-bidding games. 
Deciding the winner in a turn-based parity game is a long-standing open problem, which is known to be in \NP and co\NP but not known to be in P. 
A very simple reduction from turn-based parity games to parity discrete-bidding games was shown in~\cite{AAH21}. Moreover, the reduction outputs a bidding game with a total budget of $0$; that is, a discrete bidding game with constant sum of budgets.
Our results show that parity discrete-bidding games are in \NP and co\NP even when the sum of budgets is given in binary. One might expect that such games would be at least exponentially harder than bidding games with constant sum of budgets. But all of these classes of games actually lie in \NP and co\NP.

\printbibliography

\end{document}